\documentclass[10pt, conference]{IEEEtran}
\usepackage{amsmath,amsthm,amsfonts}
\usepackage{graphicx}
\usepackage{epsf}
\usepackage{mdwmath}
\usepackage{mdwtab}
\usepackage{fmtcount}
\usepackage{cite}
\usepackage{rotating}
\usepackage{flushend}
\usepackage{amssymb}
\usepackage{algorithmic}
\usepackage{algorithm}
\usepackage{amssymb}
\usepackage{float}
\usepackage{algorithmic}
\usepackage{algorithm}
\usepackage{bm}
\usepackage{color}
\newtheorem{thm}{Theorem}

\theoremstyle{remark}

\newcommand{\E}{{\cal{E}}}

\newcommand{\bx}{\bm{x}}

\def\R{\mathbb{R}}
\def\E{\mathbb{E}}
\def\P{\mathbb{P}}
\def\1{\bm{1}}

\def\S{\mathcal{S}}

\def \J{{\mathcal{J}}}

\IEEEoverridecommandlockouts

\title{Near-Optimal Adaptive Compressed Sensing }

\begin{document}

\author{\IEEEauthorblockN{Matthew L. Malloy, \emph{Member, IEEE} \; and \; Robert D. Nowak, \emph{Fellow, IEEE}}
 \thanks{Manuscript received May 24, 2013; revised November 7, 2013; accepted April 19, 2014. This work was partially supported by AFOSR grant FA9550-09-1-0140, NSF grant CCF-1218189 and the DARPA KECoM program.  The material in this paper was presented in part at the Asilomar Conference on Signals, Systems and Computers, Pacific Grove, CA, USA \cite{MalloyNearOptimalACS}, November, 2012.
\newline \indent M. L. Malloy and R. D. Nowak are with the Department of Electrical and Computer Engineering, University of Wisconsin, Madison, WI 53715 USA (e-mail: mmalloy@wisc.edu, nowak@engr.wisc.edu).
}
}
\maketitle

\begin{abstract} 
This paper proposes a simple adaptive sensing and group testing algorithm for sparse signal recovery.  The algorithm, termed Compressive Adaptive Sense and Search (CASS), is shown to be near-optimal in that it succeeds at the lowest possible signal-to-noise-ratio (SNR) levels, { improving on previous work in adaptive compressed sensing \cite{haupt2009adaptive, iwen2009group, iwen2011adaptive, haupt2012sequentially }.} 
Like traditional compressed sensing based on random non-adaptive design matrices, the CASS algorithm requires only $k\log n$ measurements to recover a $k$-sparse signal of dimension $n$. However, CASS succeeds at SNR levels that are a factor $\log n$ less than required by standard compressed sensing.   From the point of view of constructing and implementing the sensing operation as well as computing the reconstruction, the proposed algorithm is substantially less computationally intensive than standard compressed sensing.  CASS is also demonstrated to perform considerably better in practice through simulation.  To the best of our knowledge, this is the first demonstration of an adaptive compressed sensing algorithm with near-optimal theoretical guarantees {\em and} excellent practical performance.  This paper also shows that methods like compressed sensing, group testing, and pooling have an advantage beyond simply reducing the number of measurements or tests -- adaptive versions of such methods can also improve detection and estimation performance when compared to non-adaptive direct (uncompressed) sensing.

\end{abstract}
\section{Introduction}
Compressed sensing (CS) has had a tremendous impact on signal processing, machine learning, and statistics, fundamentally changing the way we think about sensing and data acquisition.   
Beyond the standard compressed sensing methods that gave birth to the field, the compressive framework naturally suggests an ability to make measurements in an on-line and adaptive manner.  Adaptive sensing uses previously collected measurements to guide the design and selection of the next measurement in order to optimize the gain of new information.

There is now a reasonably complete understanding of the potential advantages of adaptive sensing over non-adaptive sensing, the main one being that adaptive sensing can reliably recover sparse signals at lower SNRs  than non-adaptive sensing {(where SNR is proportional to the squared amplitude of the weakest non-zero element)}. Roughly speaking, to recover a $k$-sparse signal of length $n$, standard (non-adaptive) sensing requires the SNR to grow like $\log n$. Adaptive sensing, on the other hand, succeeds as long as the SNR scales like $\log k$.  This is a significant improvement, especially in high-dimensional regimes.
In terms of the number of measurements, both standard compressed sensing and adaptive compressed sensing require about $k\log n$ measurements.

This paper makes two main contributions in adaptive sensing.  First, we propose a simple adaptive sensing algorithm, termed Compressive Adaptive Sense and Search (CASS) that is proved to be near-optimal in that it succeeds if the SNR scales like $\log k$.   From the point of view of constructing and implementing the sensing operation as well as computing the reconstruction, the CASS algorithm is comparatively less computationally intensive than standard compressed sensing. The algorithm could be easily realized in a number of existing compressive systems, including those based on digital micromirror devices \cite{duarte2008single,studer2012compressive}.  Second, CASS is demonstrated in simulation to perform considerably better than \emph{1)} compressed sensing based on random Gaussian sensing matrices and \emph{2)}  non-adaptive direct sensing (which requires $m=n$ measurements).  To the best of our knowledge, this is the first demonstration of an adaptive sensing algorithm with near-optimal theoretical guarantees  \emph{and} excellent practical performance.  

The results presented in this paper answer the following important question.  Can sensing systems that make adaptive measurements significantly outperform non-adaptive systems in practice?  Perhaps not surprisingly, the answer is \emph{yes}, contradictory to the theme of \cite{2011arXiv1111.4646A}.  This is for two reasons; first, while the required SNR for success of adaptive procedures is at best a logarithmic factor in dimension smaller than non-adaptive procedures, this factor \emph{is} significant for problems of even modest size.  On test signals (natural and synthetic images, see Figs.  \ref{fig:eights} and \ref{fig:CameraMan2}) the CASS procedure consistently outperforms standard compressed sensing by 2-8 dB.   Second, in terms of computational complexity, adaptive sensing algorithms such as CASS require little or no computation after the measurement stage of the procedure is completed; in general, non-adaptive algorithms involve computation and memory intense optimization routines after measurements are gathered.

This work also sheds light on another relevant question. 
Can compressive sensing systems with an SNR budget 
ever hope to outperform direct non-compressive systems which measure each element of an unknown vector directly? The answer to this question is somewhat surprising.  For the task of support recovery, when the measurements are collected non-adaptively, compressive techniques are always less reliable than direct sensing with $n$ measurements.  However, \emph{adaptive} compressed sensing with just $k\log n$ measurements can be more reliable than non-adaptive direct sensing, provided the signal is sufficiently sparse.  This means that methods like compressed sensing, group testing, and pooling may have an advantage beyond simply reducing the number of measurements  -- adaptive versions of such methods can also improve detection and estimation performance.  

{
Lastly, while adaptive sensing results in improved performance, we mention a few limitations.  Making adaptive measurements can necessitate more flexible measurement hardware, incurring a cost in hardware design. In some systems, adaptive measurements could be easily realized (for example, existing compressive systems based on digital micromirror devices \cite{duarte2008single, studer2012compressive}).  In other systems, however, hardware constraints may make adaptive sensing entirely impractical.  The results in this paper assume adaptive measurements can be made and the additive Gaussian noise model is applicable; both requirements are made precise in the following section. }



\subsection{Problem Setup and Background}

Sparse support recovery in compressed sensing refers to the following problem.  Let $\bm{x} \in \mathbb{R}^n$ be an unknown signal with $k \ll n$ non-zero entries.  The unknown signal is measured through $m$ linear projections of the form:
\begin{eqnarray} \label{eqn:sensEqn}
y_i = \left\langle \bm{a}_i, \bm{x}  \right\rangle + z_i \qquad i = 1,...,m
\end{eqnarray}
where $z_i$ are i.i.d. $\mathcal{N}(0,1)$, and $\bm{a}_i\in\R^n$ are sensing vectors with a total sensing energy constraint
\begin{eqnarray} \label{eqn:sensEngr}
\sum_{i=1}^m||\bm{a}_i||^2_2 \leq M \: .
\end{eqnarray}
The goal of the support recovery problem is to identify the locations of the non-zero entries of $\bm{x}$.   
Roughly speaking, the paramount result of compressed sensing states this is possible if \emph{1)} $m$ is greater than a constant times $k \log n$ (which can be much less than $n$), and \emph{2)} the amplitude of the non-zero entries of $\bx$, and the total sensing energy $M$, aren't too small. 

This paper is concerned with relaxing the second requirement through adaptivity.  If the sensing vectors $\bm{a}_1,...,\bm{a}_m$ are fixed prior to making any measurements (the \emph{non-adaptive} setting), then a necessary condition for exact support recovery is that the smallest non-zero entry of $\bm{x}$ be greater than a constant times $\sqrt{\frac{n}{M} \log n}$ \cite{5319750, 5571873}.  
The factor of $M/n$ is best interpreted as the sensing energy per dimension, and $\log n$ is required to control the error rate across the $n$ dimensions.  Standard compressed sensing approaches based on random sensing matrices are known to achieve this lower bound while requiring on the order of $k \log n$ measurements \cite{wainwright2006sharp}.

In adaptive sensing, $\bm{a}_i$ can be a function of $y_1,...,y_{i-1}$.  The sensing vectors $\bm{a}_1,...,\bm{a}_m$  are \emph{not} fixed prior to making observations, but instead depend on previous measurements.  
In this adaptive scenario, a necessary condition for sparse recovery$^2$ is that the smallest non-zero entry exceed a constant times $\sqrt{\frac{n}{M} \log k}$ \cite{2012arXiv1206.0648C}.   
To the best of our knowledge, no proposed algorithms achieve this lower bound while also using only order  $k \log n$ measurements. 

A handful of adaptive sensing procedures have been proposed, some coming close to meeting the lower bound while requiring only order $k \log n$ measurements.   Most recently, in \cite{haupt2012sequentially}, the authors propose an algorithm that guarantees exact recovery provided the smallest non-zero entry exceeds an unspecified constant times $\sqrt{\frac{n}{M}\left( \log k + \log \log_2 \log n\right)}$, coming within a triply logarithmic factor of the lower bound.  While this triply logarithmic factor is not of practical concern, removing the suboptimal dependence on $n$ is of theoretical interest.  Reducing the leading constant, on the other hand, is of great practical concern.

\begin{table}
\centering 
\caption{Asymptotic requirements for exact support recovery$^1$ \newline $\left( n, k \rightarrow \infty, \; \frac{k}{n} \rightarrow 0 \right)$  \label{table1}}
 \smallskip
\begin{minipage}{9.4cm}
\addtocounter{footnote}{1}
\begin{tabular}{ p{.15cm}||p{3.6cm}|p{3.6cm}}
    &  \emph{non-adaptive} & \emph{adaptive} \\
  \hline
  \hline
\vspace{.3cm}  \begin{sideways} \emph{direct} \end{sideways}   &  \vspace{.01cm} \begin{enumerate} 
\item $m \geq n$  
\item $x_{\min} > \sqrt{2 \frac{n}{M} \log n}$ 
\end{enumerate} 
necessary and sufficient \newline traditional detection problem  &  \vspace{.01cm}
\begin{enumerate} 
\item $m > n$  
\item   $x_{\min} > \sqrt{2 \frac{n}{M} \log k} $ 
\end{enumerate}  
necessary \cite{malloy2011limits} \newline
 sufficient \cite{SequentialTesting}  \\
  \hline  
\vspace{.01cm} \begin{sideways} \emph{compressive} \end{sideways}   & \hspace{.1cm}  
\begin{enumerate} \footnotesize
\item $m \geq k \log n$  
\item $x_{\min} > \sqrt{C \frac{n}{M} \log n}$ 
\end{enumerate}
necessary \cite{5319750, 5571873}
\newline sufficient \cite{wainwright2006sharp}  
  & \hspace{.1cm} 
\begin{enumerate} \footnotesize
\item $m \geq k \log n$ 
\item $x_{\min} > \sqrt{C \frac{n}{M} \log k}$  
\end{enumerate}
necessary \cite{2012arXiv1206.0648C}$^2$ \newline sufficient using CASS$^3$
 \\  
  \hline
 \end{tabular}\par
 
\end{minipage}
\end{table}

 \footnotetext[\value{footnote}]{{We write $g(n) > f(n)$ as shorthand for $\lim_{n\rightarrow \infty} {g(n)}/{f(n)} > 1$.}}

\setcounter{footnote}{2}
 \footnotetext[\value{footnote}]{{Necessary to control the expected symmetric set difference over a slightly larger class of problems. See  Sec. \ref{sec:CASSdis}  and \cite[Proposition 4.2]{2012arXiv1206.0648C}.}}

\setcounter{footnote}{3}
 \footnotetext[\value{footnote}]{Sufficient for exact support recovery with Alg. \ref{alg:CASS1pass}  if ${x_i}\geq0$, $i=1,\dots,n$ or if $k \leq \log n$; sufficient to recover any fixed fraction of support for any $k$, $\bm{x}$ that satisfy condition.  See Theorems \ref{thm:kCBS} and \ref{thm:PosNeg}.  
 }

Table \ref{table1} summarizes the necessary and sufficient conditions for exact support recovery in the compressive, direct, non-adaptive and adaptive settings.  \emph{Direct} or \emph{un-compressed} refers to the setting where individual measurement are made of each component -- specifically, each sensing vector is supported on only one index.  Here, support recovery is a traditional detection problem, and the non-adaptive SNR requirement becomes apparent by considering the following.  If we measure each element of $\bm{x}$ once using sensing vectors that form an identity matrix when stacked together,  then $M=n$, and the requirement implies the smallest non-zero entry be greater than $\sqrt{2 \log n}$. It is well known that the maximum of $n$ i.i.d. Gaussians grows as $\sqrt{2 \log n}$; the smallest signal must exceed this value. 

Adaptive direct sensing has recently received a fair amount of attention in the context of sparse recovery \cite{SequentialTesting, wei2012multistage, newstadt2010adaptive, haupt2011distilled, malloy2011sequential, malloy2012sample, price2012lower, jamieson2013finding, jamieson2013lil}, and is closely related to traditional work in sequential statistical analysis.  In this non-compressive setting, the number of measurements must be at least on the order of the dimension -- i.e, $m \geq n$, as testing procedures must at a minimum measure each index once.    In this setting, a number of works have shown that scaling of the SNR as $\log k$ is necessary and sufficient for exact recovery \cite{SequentialTesting, malloy2011sequential}.   

To summarize, Table \ref{table1} highlights the potential gains from compressive and adaptive measurements.  Allowing for compressive measurements can reduce the total number of measurements from $n$ to $k \log n$, but does not relax the SNR requirement.  Allowing for adaptive measurements does not reduce the total number of measurements required; instead, it can reduce the required SNR scaling from $\log n$ to $\log k$.  Prior to this work, in the adaptive compressive case, it was unknown if $\log k$ scaling of the SNR was sufficient for support recovery. 






\subsection{Main Results and Contributions}
The main theoretical contribution of this work is to complete Table \ref{table1} by introducing a simple adaptive compressed sensing procedure that \emph{1)} requires only order $k \log n$ measurements, and \emph{2)} succeeds provided the minimum non-zero entry is greater than a constant times $\sqrt{\frac{n}{M} \log k}$, showing this scaling is sufficient.  To the best of our knowledge, this is the first demonstration of a procedure that is optimal in terms of dependence on SNR and dimension. 

Specifically, we propose a procedure termed Compressive Adaptive Sense and Search (CASS).   For recovery of non-negative signals, the procedure succeeds in exact support recovery with probability greater than $1-\delta$ provided the minimum non-zero entry of $\bm{x}$ is greater than
\begin{eqnarray} \nonumber
\sqrt{20 \frac{n}{M} \left(  \log k +\log \left( \frac{8}{\delta}\right) \right)},
\end{eqnarray}
requires exactly $m = 2k \log_2 (n/k)$ measurements,  and has total sensing energy $\sum_{i=1}^m||\bm{a}_i||^2_2 = M$. 
This result implies that as $n$ grows large, if the minimum non-zero entry exceeds $\sqrt{20 \frac{n}{M}\log k}$, the procedure succeeds in exact support recovery.  

When the signal contains both positive and negative entries, the CASS procedure guarantees the following.   On average, the procedure succeeds in recovery of at least a fraction $1-\epsilon$ of the non-zero components provided their magnitude is at least
\begin{eqnarray} \nonumber
  \sqrt{ 20 \frac{n}{M} \left( \log{ k } +2 \log\left( \frac{8}{ \epsilon} \right) \right)  }.
\end{eqnarray}
In this setting the procedure requires {less than $2k \log_2 \left(n/k\right) +8k/\epsilon$} measurements and has total sensing energy $\sum_{i=1}^m||\bm{a}_i||^2_2 = M$.


In addition to these theoretical developments, we present a series of numerical results that show CASS outperforms standard compressed sensing for problems of even modest size.  The numerical experiments reinforce the theoretical results --  as the dimension of the problem grows, performance of the CASS algorithm does not deteriorate, confirming the logarithmic dependence on dimension has been removed.   We also show side-by-side comparisons between standard compressed sensing, direct sensing, and CASS on approximately sparse signals.

\subsection{{Prior Work}} \label{sec:priorWork}

A number of adaptive compressed sensing procedures have been proposed (see \cite{iwen2009group, ji2008bayesian, iwen2011adaptive, 2012arXiv1202.0937D, haupt2012sequentially, iwen2012adaptive, 5470138} and references therein), some coming close to meeting the lower bound while requiring only order $k \log n$ measurements.   

Most closely related to the CASS procedure proposed here, from an algorithmic perspective, are the procedures of Iwen \cite{iwen2009group}, Iwen and Tewfik \cite{iwen2011adaptive, iwen2012adaptive}, and the Compressive Binary Search of Davenport and Arias-Castro \cite{2012arXiv1202.0937D}. Much like CASS, these procedures rely on repeated bisection of the support of the signal.  The first proposal of a bisecting search procedure applied to the adaptive compressed sensing problem, to the best of our knowledge, was \cite{iwen2009group}.  The 
theoretical guarantees of this work when translated to the setting studied here ensure exact support recovery provided the non-zero entries of $\bm{x}$ are greater than a constant times $\sqrt{\frac{n}{M} ( \log^2 k + \log^2 \log n ) }$ (see \cite[Theorem 2]{iwen2012adaptive} and more succinctly, \cite[Theorem 2]{iwen2011adaptive}).

Most recently, in \cite{haupt2012sequentially}, the authors propose an algorithm termed Sequential Compressed Sensing.  The procedure is based on random sensing vectors with sequentially masked support, allowing the procedure to focus sensing energy on the suspected non-zero components. The procedure guarantees exact recovery provided the smallest non-zero entry exceeds an unspecified constant times $\sqrt{\frac{n}{M}\left( \log k + \log \log_2 \log n\right)}$, coming within in a triply logarithmic factor of the theoretical bound lower bound.

While the bisection approach used by CASS is suggested in a number of other works, the distinguishing feature of the 
the procedure presented here is the allocation of the sensing energy. In the case of adaptive compressed sensing, roughly speaking, both procedures in \cite{iwen2009group}, and later in \cite{2012arXiv1202.0937D}, suggest an allocation of sensing energy across the bisecting steps of the procedure that ensure the probability a mistake is made at any step is the same.   The CASS procedure instead uses less sensing energy for initial steps, and more sensing energy at later steps, reducing the the probability the procedure makes an error from step to step.  This allocation of the sensing energy removes any sub-optimal dependence on dimension from the SNR requirement; the CASS procedure succeeds provided the minimum non-zero entry exceeds $\sqrt{ 20 \frac{n}{M}  \log k} $.  

{Compressed sensing can be viewed as a variant of non-adaptive group testing (see \cite{atia2012boolean, gilbert2008group} for a detailed discussion).  In traditional group testing \cite{du1993combinatorial}, binary valued measurements indicate if a group of items contains one more defectives, distinct from the real valued measurement model in (1).  While the measurement models are different, the CASS procedure attempts to emulate adaptive group testing by forming partial sums of the sparse signal in an adaptive, sequential manner. In the traditional group testing literature, CASS is perhaps most similar to the \emph{generalized binary splitting} procedure of \cite{hwang1972method}. Like generalized binary splitting, CASS first isolates defective items by measuring groups of size approximately $n/k$.   The procedures deviate beyond this; the main distinguishing feature of CASS is the allocation of the sensing energy across passes required to cope with additive Gaussian noise.  To the best of our knowledge, analogous allocations of have not been proposed to cope with noise in adaptive group testing as noisy group testing has remained largely unexplored \cite{atia2012boolean}.
This highlights another contribution of CASS -- although we focus on the compressed sensing framework, the key idea (the allocation of the sensing energy) could be adapted to other measurement models, such as a binary symmetric noise model in group testing.  In that case, CASS would suggest an allocation of measurements across stages to compensate for a fixed probability of error per measurement.}




\subsection{Notation}
Notation adopted throughout, in general, follows convention.  Bold face letters, $\bm{x}$, denote vectors (both random and deterministic), matrices are denoted with capital letters, for example, $A$, while calligraphic font, such as $\S$, denotes a set or an event.  The indicator function $ \mathrm{I}_{\{\mathcal{E}\}}$ is equal to one if the event $\mathcal{E}$ occurs, and equals zero otherwise.  The vector $\bm{1}_{\{\S\}}$, for any support set $\S$, is a vector of ones over $\S$, and $0$ elsewhere. Expectation is denoted $\mathbb{E} [ \cdot] $ and the probability of an event $\mathbb{P}(\cdot)$.


\section{Problem Details} \label{sec:details}
Consider the canonical compressed sensing setup in which a sparse signal $\bm{x} \in \mathbb{R}^n$ supported on $\mathcal{S} \subset \{1,\dots, n\}$, $|\S| = k$, is measured through $m$ linear projections as in (\ref{eqn:sensEqn}).  In matrix-vector form, 
\begin{eqnarray} \nonumber
\bm{y} = A \bm{x} + \bm{z}
\end{eqnarray}
where $\bm{z} \in \mathbb{R}^m \sim \mathcal{N}(0,I)$ and $ A \in\R^{m \times n}$ is the sensing matrix with total sensing energy constraint
\begin{eqnarray} \label{eqn:sensEngr}
||A||_{\mathrm{fro}}^2 \leq M. \: 
\end{eqnarray}
The total sensing energy constraint reflects a limit on the total amount of energy or time available to the sensing device, as in \cite{haupt2012sequentially, 2012arXiv1202.0937D}.  Importantly, it results in a decoupling between the energy constraint and the number of measurements, allowing for fair comparison between procedures that take a different number of measurements. 
In adaptive sensing the rows of $A$, denoted $\bm{a}_i^T$, are designed sequentially and adaptively based on prior observations, $y_1, \dots, y_{i-1}$, as opposed to the non-adaptive setting in which $\bm{a}_i$, $i=1,\dots,m$ are fixed a-priori.  

The error metrics considered are the probability of exact support recovery, or the family wise error rate,
\begin{eqnarray} \nonumber
 \P\left(\widehat \S \neq \S \right)
\end{eqnarray}
and the symmetric set difference  
\begin{eqnarray} \nonumber
d(\widehat{\S},S) = \left \vert \{\S \setminus \widehat{\S} \} \cup \{ \widehat{\S} \setminus \S \} \right \vert
\end{eqnarray}
where $\widehat \S$ is an estimate of $\S$.  Both depend on the procedure, $n$, $k$, $M$ and the amplitude of the non-zero entries of $\bm{x}$. {  In particular, we quantify performance in terms of the minimum amplitude of the non-zero  elements of $\bm{x}$, given as
\begin{eqnarray} \nonumber
x_{\min} = \min_{i\in \S}  |x_i|.
\end{eqnarray}  
The the SNR is defined as the amplitude of the smalled non-zero entry squared, scaled by the sensing energy per dimension:
\begin{eqnarray} \nonumber
\mathrm{SNR} = \frac{x_{\min}^2 M }{n}.
\end{eqnarray}
}

In general, no assumption is made on the support set $\S$ other than $|\S|= k$.  Theorem 2 requires the sparse support be chosen uniformly at random over all possible cardinality $k$ support sets, which can be avoided by randomly permuting the support before running the algorithm.   

Lastly, throughout the paper we assume that both $n$ and $k$ are powers of two for simplicity of presentation.  This can of course be relaxed, at the cost of increased constants.

\section{Compressive Adaptive Sense and Search} \label{sec:CASS}

Conceptually, the CASS procedure operates by initially dividing the signal into a number of partitions, and then using compressive measurements to test for the presence of one or more non-zero elements in each partition.  The procedure continues its search by bisecting the most promising partitions, with the goal of returning the $k$ largest components of the vector.  This compressive sensing and bisecting search approach inspired the name of the algorithm.  

\subsection{The CASS Procedure}
The CASS procedure is detailed in Alg.  \ref{alg:CASS1pass}.  The procedure requires three inputs: \emph{1)} $k$, the number of non-zero coefficients to be returned by the procedure, \emph{2)} $M$, the total budget, and $\emph{3)}$ and an initial scale parameter $\epsilon$, which is in most practical scenarios is set to $\epsilon = 1$. 

 Since the procedure is based on repeated bisections of the signal, it is helpful to define the dyadic subintervals of $\{1,\dots, n\}$.   The dyadic subintervals,  $\J_{j,\ell} \subset \{1,\dots, n\}$, are given by
\begin{eqnarray}  \nonumber
\J_{j,\ell} &=& \left\{  \frac{(\ell-1)n}{ 2^j }+1, \dots ,\frac{\ell n}{2^j}  \right\},  \\ \nonumber
 && j = 0,1,\dots,\log_2 n, \quad \ell = 1,\dots, 2^j  
\end{eqnarray}
where $j$ indicates the scale and $\ell$ the location of the dyadic partition. 

The procedure consists of a series of $s_0$ steps.  Conceptually, on the first step, $s=1$, the procedure begins by dividing the signal into $\ell_0$ partitions, where $\ell_0$ is the smaller of $n$ or the next power of two greater than $4k/\epsilon$.  These initial partitions are given by the dyadic sub-intervals of $\{1,\dots,n\}$ at scale $\log_2 (\ell_0)$:
\begin{eqnarray} \nonumber
\J_{\log_2 (\ell_0), \ell} \qquad \ell = 1,\dots, \ell_0.
\end{eqnarray}
The procedure measures the signal with $\ell_0$ sensing vectors, each with support over a single dyadic sub-interval.
The procedure then selects the largest $k$ measurements in absolute value.  These $k$  measurements are used to define the support of the sensing vectors used on step two of the procedure.  More specifically, the supports of the sensing vectors corresponding to the $k$ largest measurements are bisected, giving $2k$ support sets.  These $2k$ support sets define the support of the sensing vectors on step $s=2$.  

The procedure continues in the fashion, taking $2k$ measurements on each step (after the initial step, which requires $\ell_0$ measurements) and bisecting the support of the $k$ largest measurements to define the support of the sensing vectors on the next step. 
On the final step, $s = s_0$, where $s_0 = \log_2 \frac{n}{\ell_0} +1$, the support of the sensing vectors consists of a single index; the procedure returns the support of the sensing vectors corresponding to the $k$ largest measurements as the estimate of $\S$ and terminates.  Additionally, estimates of the values of the non-zero coefficients are returned based on the measurements made on the last step.

One of the key aspects of the CASS algorithm is the allocation of the measurement budget across the steps of the procedure.  On step $s$, the amplitude of the non-zero entries of the sensing vectors is given as 
\begin{eqnarray} \label{eqn:sensAmp}
 \sqrt{\frac{Ms}{\gamma n} } 
\end{eqnarray}
where  $\gamma$ is an internal parameter of the procedure used to ensure the total sensing energy constraint in (\ref{eqn:sensEngr}) is satisfied with equality.  Specifically, 
\begin{eqnarray} \nonumber
\gamma = {1} + \frac{4k}{\ell_0} \sum_{s=2}^{s_0} s 2^{-s}.
\end{eqnarray}
This particular allocation of the sensing energy gives the CASS procedure its optimal theoretical guarantees.  While the amplitude of each sensing vector grows in a polynomial manner with each step, the support of the sensing vectors decreases geometrically; this results in a decrease in the sensing energy across steps, but an increase in reliability. { In essence, the allocation exploits the fact that an increase in reliability comes at a smaller cost for steps later in the procedure.}

Details of the CASS procedure are found in Alg. \ref{alg:CASS1pass}, and performance is quantified in the following section.

\begin{algorithm}[h]
\caption{\hspace{1cm} CASS \hspace{.2cm} }
\begin{algorithmic}  \label{alg:CASS1pass}
\STATE{input: number of terms in approximation $k$, \\
\hspace{.85cm} budget $M$, \\
\hspace{.85cm} initial scale parameter $\epsilon \in (0,1]$ -- (default $\epsilon = 1$)}
\STATE{initialize: $\ell_0 = \min \left\{ 4 \cdot 2^{\lceil\log_2(k/\epsilon) \rceil} ,n  \right\}$ partitions \\ 
$\qquad \qquad s_0 = \log_2 \frac{n}{\ell_0} +1$ steps \\
$ \qquad \qquad \gamma = 1 + 4k/\ell_0 \sum_{s=2}^{s_0} s 2^{-s}  $  \\
$ \qquad \qquad \mathcal{L} = \{1,2,\dots, \ell_0 \}$ \\
}

\FOR{$s = 1, \dots, s_0$}
\FOR{$\ell \in \mathcal{L}$}
\STATE{ $\bm{a}_{s,\ell} = \sqrt{\frac{Ms}{\gamma n}} \; \bm{1}_{\left\{\J_{\log_2 (\ell_0) -1+s,\ell}\right\}}$ }
\STATE{{\bf measure}:  $y_{s,\ell} = \langle \bm{a}_{s,\ell},\bm{x}\rangle + z_{s,\ell} \qquad z_{s,\ell} \overset{iid}{\sim} \mathcal{N}(0,1) $ }
\ENDFOR
\STATE{let  $\ell_1, \ell_2, \dots,\ell_{|\mathcal{L}|}$ be such that  \\ $\qquad \qquad |y_{s,\ell_1}|\geq |y_{s,\ell_2}|\geq \dots \geq |y_{s,\ell_{|\mathcal{L}|}} |$}
\IF{$s \neq s_0 $}
\STATE{$\mathcal{L} = \left\{2 \ell_1-1, 2 \ell_1, 2 \ell_2-1, 2  \ell_2,  \dots , 2 \ell_k-1, 2 \ell_k \right\}$}
\ELSE
\STATE{$\widehat{\S} = \left\{\ell_1, \ell_2, \dots, \ell_k   \right\} $ 
\\ $\widehat{x}_{i} = y_{s,i} \cdot \sqrt{\frac{\gamma n}{Ms}} $ for all $i \in \widehat{S}$ }
\ENDIF
\ENDFOR
\STATE{{\bf output:} $\widehat{\S}$ ($k$ indices), \indent estimates $\widehat{x}_i$ for $i \in \widehat{\S}$}
\end{algorithmic}
\end{algorithm}


\subsection{Theoretical Guarantees}
The theoretical guarantees of CASS are presented in Theorems \ref{thm:kCBS} and \ref{thm:PosNeg}.  Theorem \ref{thm:kCBS} presents theoretical guarantees for non-negative signals, while Theorem \ref{thm:PosNeg} quantifies performance when the signal is both positive and negative.  

\begin{thm} \label{thm:kCBS} Assume $x_i \geq 0$ for all $i = 1,\dots,n$. Set $\epsilon = 1$.  For $\delta \in (0,1]$,
Alg. \ref{alg:CASS1pass} has $\P(\widehat{\S} \neq \S) \leq \delta$ provided 
\begin{eqnarray} \nonumber
x_{\min} \geq \sqrt{20 \frac{n}{M} \left(  \log k  + \log \left( \frac{8}{\delta}\right)  \right)},
\end{eqnarray}
has total sensing energy $||A||_{\mathrm{fro}}^2 = M$ and uses $m = 2k \log_2 \frac{n}{k}$ measurements.
\end{thm}
\begin{proof}
See Appendix A.
\end{proof}

The theoretical guaranties of Theorem \ref{thm:kCBS} do not apply to recovery of signals with positive and negative entries, as scenarios can arise where two non-zero components can cancel when measured by the same sensing vector.  In order to avoid this effect, the \emph{initial scale parameter}, $\epsilon$, can be set to a value less than one.  Doing so increases the number of partitions that are created on the first step, but reduces the occurrence of positive and negative signal cancellation.  Theorem \ref{thm:PosNeg} bounds the expected fraction of the components that are recovered when the signal contains positive and negative entries.




\begin{thm} \label{thm:PosNeg} 
Assume $\mathcal{S}$ is chosen uniformly at random.  For any  $\epsilon \in (0,1]$
Alg. \ref{alg:CASS1pass} has $\E[d(\widehat{\S}, \S)] \leq k \epsilon $
 provided 
\begin{eqnarray} \nonumber
x_{\min} \geq    \sqrt{ 20 \frac{n}{M} \left( \log{ k } +2 \log\left( \frac{8}{ \epsilon} \right) \right)  },
\end{eqnarray}
has total sensing energy $||A||_F^2 = M$ and has
\begin{eqnarray} \nonumber
m \leq \frac{8k}{\epsilon} + 2k \log_2 \left(\frac{n}{k}\right)
\end{eqnarray}
measurements. { Moreover, if $k/\epsilon$ is a power of two, the procedure uses exactly $4k/\epsilon + 2k \log_2( n \epsilon/4k)$ measurements.}
\end{thm}
\begin{proof}
See Appendix B.
\end{proof}

\subsection{Discussion} \label{sec:CASSdis}

Theorem \ref{thm:PosNeg} implies conditions under which a fraction $1-\epsilon$ of the support can be recovered in expectation.  Provided $\epsilon \geq 1/\log n$, the procedure requires order $k \log n$ measurements.  In scenarios where one wishes to control $\mathbb{P}(\widehat{\S} \neq \S)$, $\epsilon$ can be set to satisfy $\epsilon < 1/k$.  In this case, while the SNR requirement remains essentially unchanged, the procedure would require on the order of $k^2$ measurements. In this sense, the compressive case with positive and negative non-zero values is more difficult, a phenomena noted in other work \cite{arias2012detecting}.  Note that when $ k \geq \log n$, the theoretical guarantee of Theorem \ref{thm:PosNeg} adapted to control the family wise error rate requires more than order $k \log n$ measurements.

The minimax lower bound of \cite[Proposition 4.2]{2012arXiv1206.0648C} states the following (over a slightly larger class of signals with sparsity $k$, $k-1$, and $k+1$, with $k<n/2$): if $\E[d(\widehat{\S}, \S)] \leq \epsilon $ then necessarily
\begin{eqnarray} \nonumber
 x_{\min} \geq \sqrt{2\frac{n}{M} \left(\log k + \log\frac{1}{2\epsilon} -1\right)}.
\end{eqnarray}
Since the condition $\mathbb{P}(\widehat{\S} \neq \S) \leq \delta$ in Theorem \ref{thm:kCBS} implies $\E[d(\widehat{\S}, \S)] \leq 2 k \delta$ { (which is readily shown as the procedure always has $d(\widehat{\S}, \S) \leq 2k$, and has $d(\widehat{\S}, \S)=0$ when $\widehat{\S} = \S$)  the upper and lower bounds can be compared directly by setting $\delta = \epsilon/2k$. In doing so, Theorem \ref{thm:kCBS} implies that CASS succeeds with $\E[d(\widehat{\S}, \S)] \leq \epsilon $ provided 
\begin{eqnarray} \nonumber
x_{\min} \geq \sqrt{20\frac{n}{M} \left(2\log k + \log\frac{16}{\epsilon}\right)}.
\end{eqnarray}
In the same manner, Theorem \ref{thm:PosNeg} can be directly compared.  Doing so shows the theorems are tight to within constant factors.  
}

The splitting or bisecting approach in CASS is a common technique for a number search problems, including adaptive sensing.   As discussed in Sec. \ref{sec:priorWork}, to the best of our knowledge, the first work to propose a bisecting search for the compressed sensing problem was \cite{iwen2009group}.  When translated to the framework studied here, the authors in essence allocate the sensing vectors on each bisecting step to be proportional to $\sqrt{M/n}$, as opposed to increasing with the index of the step as in (\ref{eqn:sensAmp}). Likewise, the \emph{compressive binary search} (CBS) algorithm \cite{2012arXiv1202.0937D} also employs sensing vectors that are proportional to $\sqrt{M/n}$. The CBS procedure finds a single non-zero entry with vanishing probability of error as $n$ gets large provided $x_{\min} \geq \sqrt{8 \frac{n}{M} \log \log_2 n}$, where $x_{\min}$ is the amplitude of the non-zero entry.  The authors of \cite{2012arXiv1202.0937D} rightly question whether the $\log\log_2 n$ term is needed.  Thm. \ref{thm:kCBS} answers this question, removing the doubly logarithmic term and coming within a constant factor of the lower bound in \cite{2012arXiv1202.0937D}.  A more direct comparison to the work in \cite{2012arXiv1202.0937D} is found in \cite{2012arXiv1203.1804M}.

Increasing the number of partitions to further isolate sparse entries is an idea suggested in \cite{iwen2009group, iwen2012adaptive}.  In particular, the authors suggest a procedure that creates a large random partitions based on primes, essentially decreasing the support of the sensing vector such that each non-zero element is isolated with arbitrarily high probability.

\section{Numerical Experiments} \label{sec:Numerical}

This section presents a series of numerical experiments aimed at highlighting the performance gains of adaptive sensing.  In all experiments, the CASS procedure was implemented {in the sparsifying basis} according to Alg. \ref{alg:CASS1pass} with inputs as specified.  The signal to noise ratio is defined as
\begin{eqnarray} \nonumber
\mbox{SNR (dB) } = 10 \log_{10} \left(\frac{x_{(k)}^2 M}{n}\right)
\end{eqnarray}
where $x_{(k)}$ is the amplitude of the $k$th largest element of $\bm{x}$.
This definition reflects the SNR of measuring the $k$th largest element using a fraction $1/n$ of the total sensing energy. In all experiments, Gaussian noise of unit variance is added to each measurement, as specified in (\ref{eqn:sensEqn}).

Fig.  \ref{fig:oneSparN}  shows empirical performance of CASS for 1-sparse recovery, and, for comparison, \emph{1)} traditional compressed sensing with orthogonal matching pursuit (OMP) using a Gaussian ensemble with normalized columns and \emph{2)} direct sensing.  As the dimension of the problem is increased, notice that performance of CASS remains constant, while the error probability of both OMP and direct sensing increase.  Note for non-adaptive methods, in the one sparse case OMP is equivalent to the maximum likelihood support estimator, and is thus optimal amongst non-adaptive recovery algorithms.  The total sensing energy is the same in all cases, allowing for fair comparison.  For sufficiently large $n$, CASS outperforms direct sensing.

\begin{figure}[]
\centerline{\includegraphics[width=9.5cm]{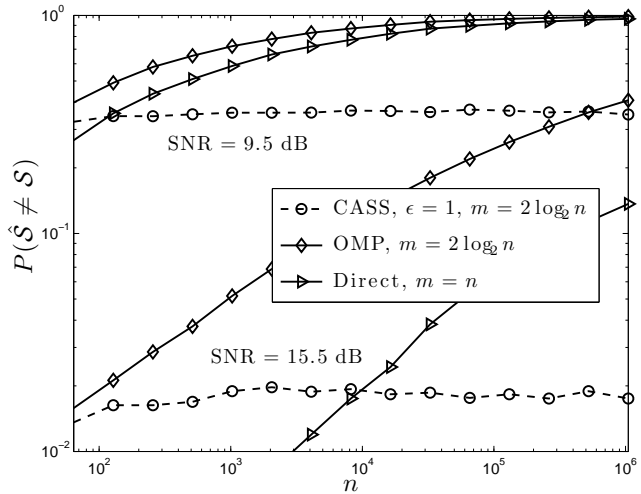}}
\caption{Recovery of one sparse signal ($k=1$).  Empirical performance of CASS (Alg. \ref{alg:CASS1pass}), traditional compressed sensing with orthogonal matching pursuit (OMP), and direct sensing as a function of $n$, for $\mathrm{SNR} (\mbox{dB}) = 10 \log_{10} (x_{\min}^2 M/n) = 9.5$ dB and $\mathrm{SNR} = 15.5$ dB. 10,000 trials for each $n$. \label{fig:oneSparN}}
\end{figure}


Figs. \ref{fig:kSparSNR} and \ref{fig:KSparN} show performance of CASS for $k$-sparse signals with equal magnitude positive and negative non-zero entries.  Performance is in terms of empirical value of $d(\widehat \S, \S)$ averaged over the trials.  The support of the test signals was chosen uniformly at random from the set of $k$-sparse signals, and the non-zero entries were assigned an amplitude of $+x_{\min}$ or $-x_{\min}$ with equal probability. 
The cancellation effect results in an error floor of around $d(\widehat{\S},\S)/(2k) = 0.10$; this is greatly reduced when the scale parameter is set to $\epsilon = 1/8$, clearly visible in both Fig. \ref{fig:kSparSNR} and Fig. \ref{fig:KSparN}.  The procedure is compared against traditional compressed sensing using a random Gaussian ensemble, recovered with LASSO and a regularizer tuned to return a $k$-sparse signal (using SpaRSA, \cite{wright2009sparse}).  Performance of traditional compressed sensing with $m = 2k \log(n/k)$, fairly compared against CASS with $\epsilon =1$, is shown in both figures.  Traditional compressed sensing with {$m = 4k/\epsilon + 2k \log_2 (n \epsilon / 4k)$} is also shown, which is fairly compared to CASS with $\epsilon = 1/8$.

\begin{figure}[]
\centerline{\includegraphics[width=9.5cm]{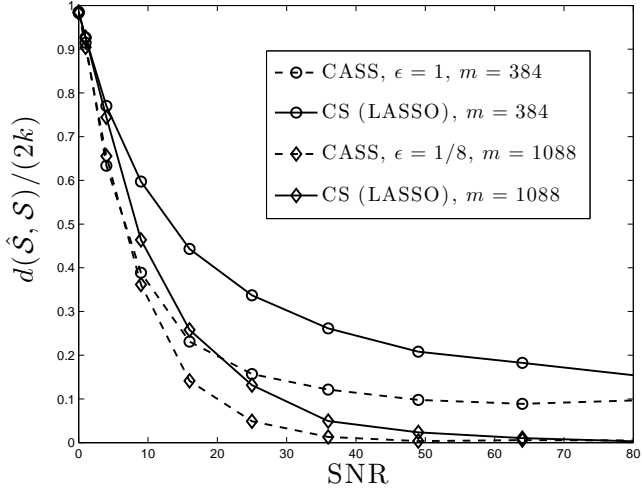}}
\caption{Empirical performance of CASS (Alg. \ref{alg:CASS1pass}) with non-zero entries of $\bx$ positive and negative of equal amplitude as a function of SNR compared against traditional compressed sensing with Gaussian ensemble and recovery using LASSO,  $k = 32$, $n = 2048$.  $\mathrm{SNR} = x_{\min}^2 M/n$. 100 trials for each SNR (linear units).  \label{fig:kSparSNR}}
\end{figure}

\begin{figure}[]
\centerline{\includegraphics[width=9.5cm]{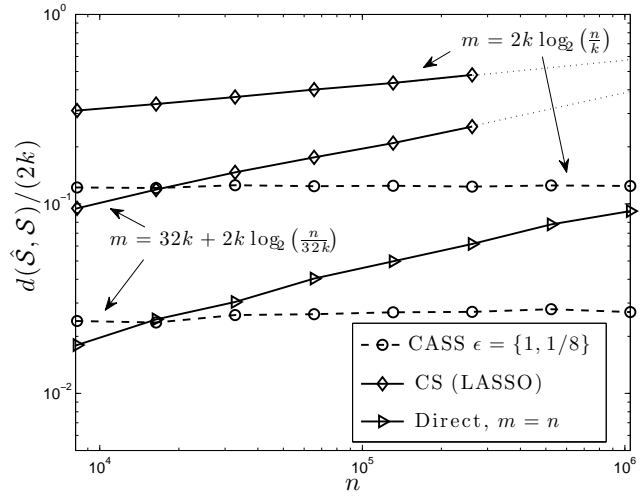}}
\caption{Non-zero entries of equal magnitude, positive and negative, as a function of $n$.  $k = 32$. CASS Alg. \ref{alg:CASS1pass} with $\epsilon = 1$, $m = 2 k \log_2 (n/k)$, and $\epsilon = 1/8$, $m =32k+2k\log_2 (\frac{n}{32k})$.  LASSO using Gaussian ensemble, $m = 2 k \log_2 (n/k)$, $m =32k+2k\log_2 (\frac{n}{32k}) $.   SNR $ = 15.5$ dB. 1000 trials for each $n$.  LASSO was not evaluated for $n\geq 2^{19}$ because of computational intensity. \label{fig:KSparN}}
\end{figure}

In Fig. \ref{fig:KSparN} notice that performance of CASS remains constant as $n$ becomes large.  For sufficiently large $n$, CASS outperforms direct sensing, highlighting a major advantage of adaptive sensing: for support recovery, while standard compressed sensing does not outperform direct sensing in terms of dependence on SNR, for sufficiently sparse problems, CASS does.    As solving LASSO with a dense sensing matrix for large problem sizes becomes computationally prohibitive, performance was only evaluated up to $n = 2^{18}$.  For all dimensions evaluated, compressed sensing using LASSO was inferior to CASS and direct sensing.

 \begin{figure}[H]
\begin{tabular}{l}
\centerline{\includegraphics[width=10.1cm]{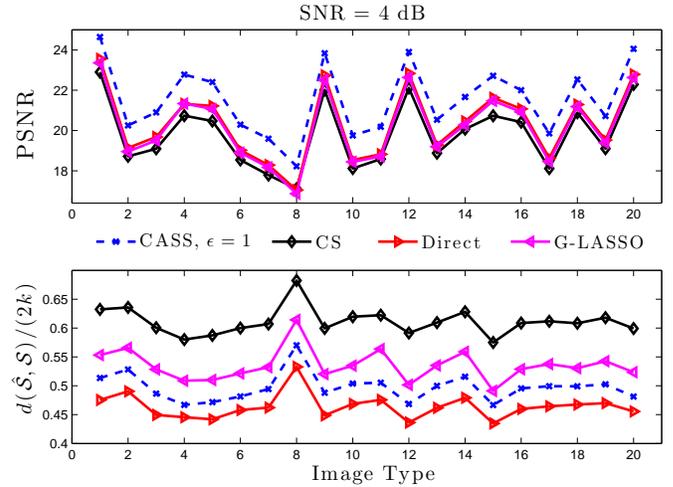}} \vspace{.3cm} \\
\centerline{\includegraphics[width=10.1cm]{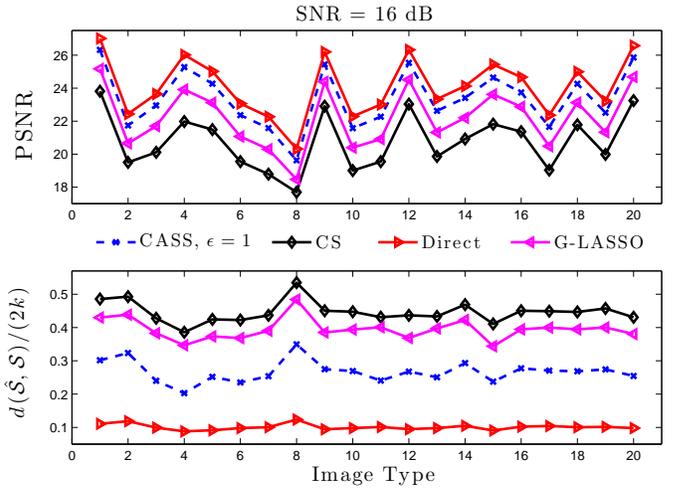}} \\
\end{tabular}
\caption{\label{fig:MSdb} Performance of CASS, standard compressed sensing with LASSO and Gaussian ensemble, Group LASSO of \cite{rao2011convex} (G-LASSO) and Gaussian ensemble, which exploits signal structure, and direct sensing on the Microsoft Research Object Class Recognition database \cite{4028553} for two SNRs.    The database consists of 20 image classes, with approximately 30 images per class.  The images are approximately sparse in the Haar basis.  The total sensing energy was the same for all methods. $n = 128^2$.  The results are averaged over each class.  $k = 512$, $\epsilon =1$.    }
\end{figure}

Fig. \ref{fig:MSdb} shows results of running CASS on the 20 image classes in the Microsoft Research Object Class Recognition database \cite{4028553}.  Each class consists of approximately 30 images of size $128 \times 128$.  Results were averaged over all images in a particular class.  Direct sensing, standard compressed sensing using a Gaussian Ensemble and the group LASSO  (G-LASSO)  of \cite{rao2011convex, jacob2009group} which exploits signal structure, are included.   G-LASSO is known to be one of the best non-adaptive methods for recovery of natural images, making it a natural benchmark for comparison (see \cite[Fig. 5]{rao2012correlated}).   The total sensing energy was the same in all cases, and all methods used $m = 5120$ measurements (except direct sensing which requires $m = 16,384$ measurements).  Each algorithm was tuned to output a $512$ coefficient approximation.   
CASS universally outperforms both G-LASSO and  standard CS on all image sets and SNRs evaluated.

Fig. \ref{fig:eights} shows recovery of a test image from \cite{FrankAsuncion:2010} for three SNRs.  The image is \emph{approximately} sparse in the Daubechies wavelet domain (`db4' in MATLAB).   Alg. \ref{alg:CASS1pass} was run with $k = 256$ as input and compared against LASSO with a random Gaussian ensemble, and the regularizer was tuned to return $k = 256$ elements  (again using SpaRSA, \cite{wright2009sparse}).  Using traditional compressed sensing, at an SNR of $ -8$ dB, the reconstructed image is essentially un-recognizable; at the same SNR, the image recovered with CASS is still very much recognizable.   The peak-signal-to-noise ratio (PSNR) for images on $[0,1)$ is defined as 
\begin{eqnarray} \nonumber
\mbox{PSNR (dB) } = -10 \log_{10} \left( \frac{|| \bm{x} -\widehat{\bm{x}}||_2^2}{n } \right)
\end{eqnarray}
where $\widehat{\bm{x}}$ is an estimate of the image.  

The experiments of Fig. \ref{fig:eights} were repeated 10 times, and the empirical results averaged over the 10 trials are shown in Table \ref{table:Digits}.  In terms of PSNR, CASS outpeforms compressed sensing by 4-8 dB.  The runtime of the simulation, which consists of the entire simulation time, including both measurement and recovery, is included.  While the compressed sensing experiments were not optimized for speed, it seems unlikely any recovery method using dense Gaussian matrices would approach the runtime of CASS -- standard compressed sensing runtime was on the order of a minute, while runtime of CASS was approximately $0.1$ seconds.  

Fig. \ref{fig:CameraMan2} shows side-by-side performance of CASS, traditional compressed sensing with a dense Gaussian sensing matrix and LASSO tuned to return $k = 2048$ non-zero components, and direct sensing on a natural image.   In terms of PSNR, CASS outperforms compressed sensing by 2-4 dB.  In terms of symmetric set difference, CASS substantially outperforms compressed sensing on all SNRs evaluated.

\begin{table}
\centering 
\caption{Handwritten Eight, 10 trials, $n=256^2$, $k =256$, $\epsilon =1$ \label{table:Digits}}
\begin{tabular}{ p{2.1cm}||p{1.7cm}|p{1.6cm}|p{1.5cm}}
& CASS (Alg. \ref{alg:CASS1pass}) & CS (LASSO) & Direct \\ \hline \hline
 \underline{$\mathrm{SNR} = -2$ dB }  & & & \\ 
 \hfill $d(\widehat{\S}, \S)$ &   120.0 & 153.2 & 131.0  \\
\hfill PSNR (dB) & $16.88$ & $8.43$ & $12.71$ \\
\hfill runtime (s) & 0.11 & 118.3  & 0.04
\\ \hline
\underline{$\mathrm{SNR} = -8$ dB}  & & & \\ 
 \hfill $d(\widehat{\S}, \S)$ &   152.0 & 202.0 & 163.8 \\
\hfill PSNR (dB)  & 12.91 & 4.40 &  6.15 \\
\hfill runtime (s) & 0.11 & 65.0 & 0.04
\\ \hline
\underline{$\mathrm{SNR} = -14$ dB}  & & & \\ 
 \hfill $d(\widehat{\S}, \S)$ &  236.4 & 379.6 & 301.6 \\
\hfill PSNR (dB) & 6.69 & 2.52 & -1.76  \\
\hfill runtime (s) & 0.11 & 54.2 & 0.03 \\
\end{tabular}
\end{table}

\begin{figure*}[]
\centerline{ 
\begin{tabular}{p{.6cm}p{5.2cm}p{5.2cm}p{5.2cm}}
& \centerline{CASS (Alg. \ref{alg:CASS1pass})}  & \centerline{CS (LASSO)} & \centerline{Direct} \\ 
\begin{sideways}  \hspace{1.4cm}  $\mathrm{SNR} = -2$ dB \end{sideways}  & 
\centerline{\includegraphics[width=2.5in]{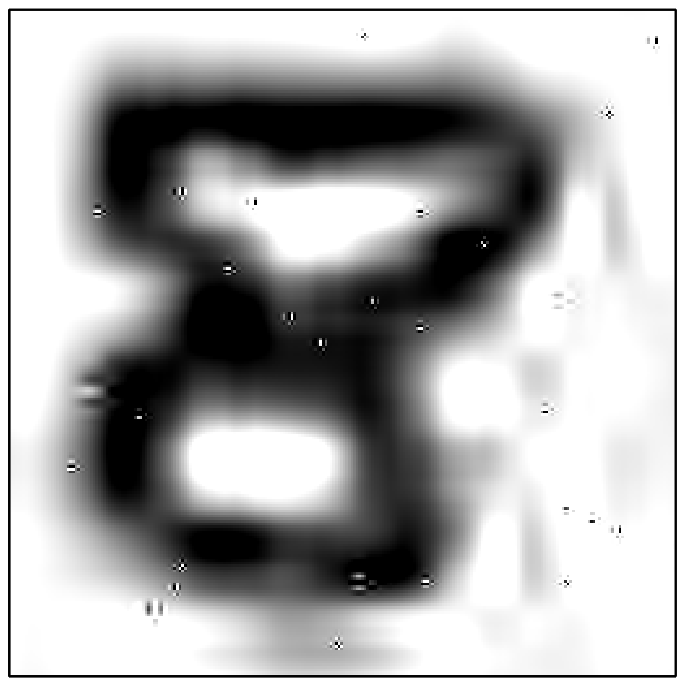} }  \footnotesize\vspace{-.7cm}\centerline{$d(\widehat{\S},\S) = 112 \quad $  PSNR $= 17.5$ dB} \centerline{simulation time $0.13$s} &    
\centerline{\includegraphics[width=2.5in]{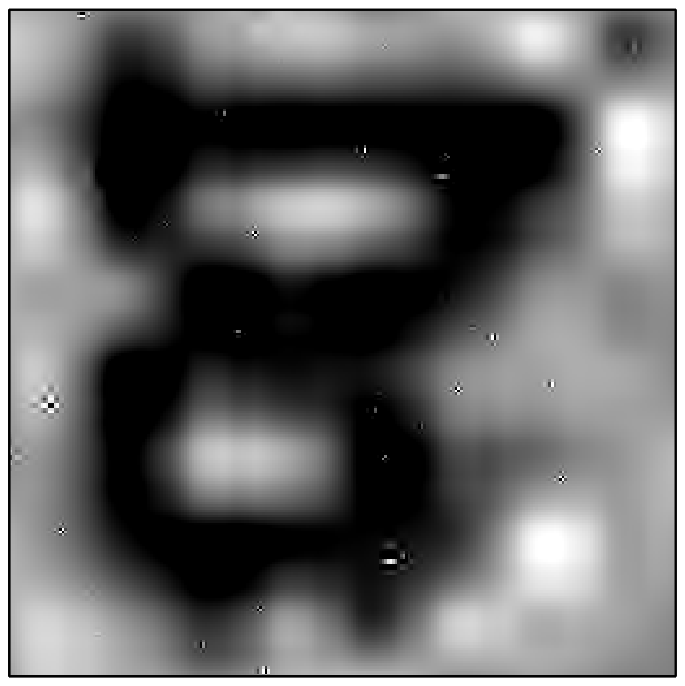} } \footnotesize\vspace{-.7cm}\centerline{$d(\widehat{\S},\S) = 152 \quad $  PSNR $= 8.6$ dB} \centerline{simulation time $69$s} & 
\centerline{\includegraphics[width=2.5in]{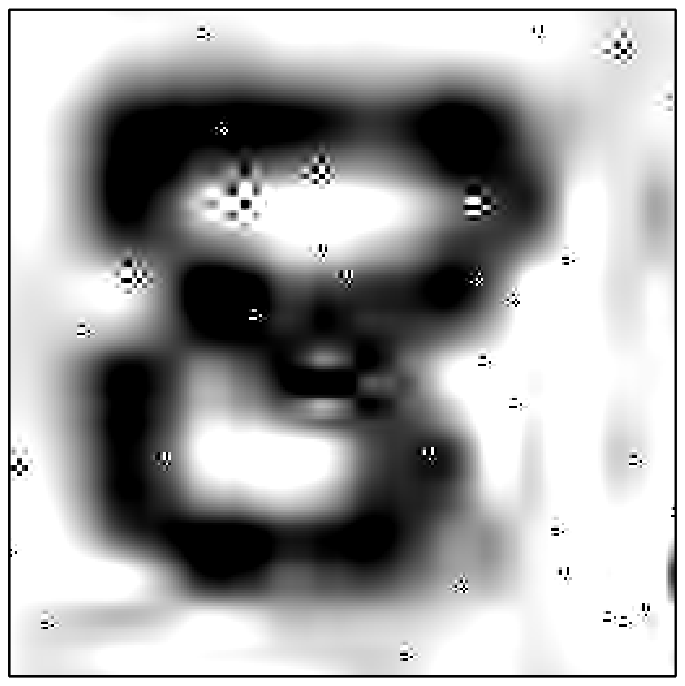} } \footnotesize\vspace{-.7cm}\centerline{$d(\widehat{\S},\S) = 134 \quad $  PSNR $=  12.3$ dB} \centerline{simulation time $0.05$s}
\\
\begin{sideways}  \hspace{1.4cm}  $\mathrm{SNR} = -8$ dB \end{sideways}  & 
\centerline{\includegraphics[width=2.5in]{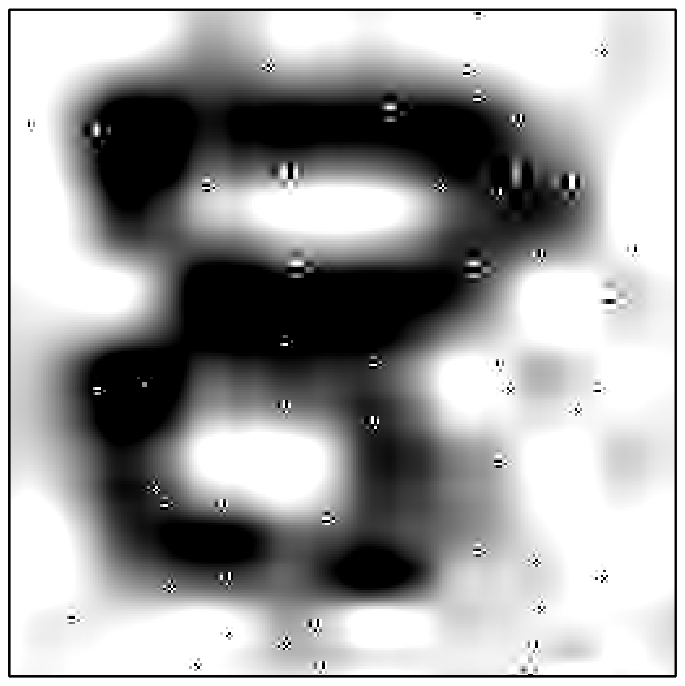}} \footnotesize\vspace{-.7cm}\centerline{$d(\widehat{\S},\S) = 156 \quad $  PSNR $= 12.9$ dB} \centerline{simulation time $0.10$s}  &     
\centerline{\includegraphics[width=2.5in]{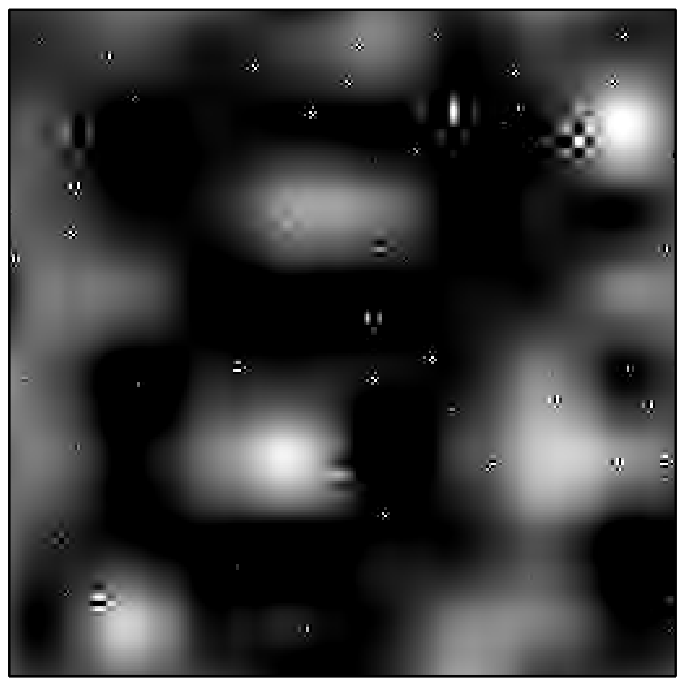} } \footnotesize\vspace{-.7cm}\centerline{$d(\widehat{\S},\S) = 198 \quad $  PSNR $= 4.6 $ dB} \centerline{simulation time $71$s} & 
\centerline{\includegraphics[width=2.5in]{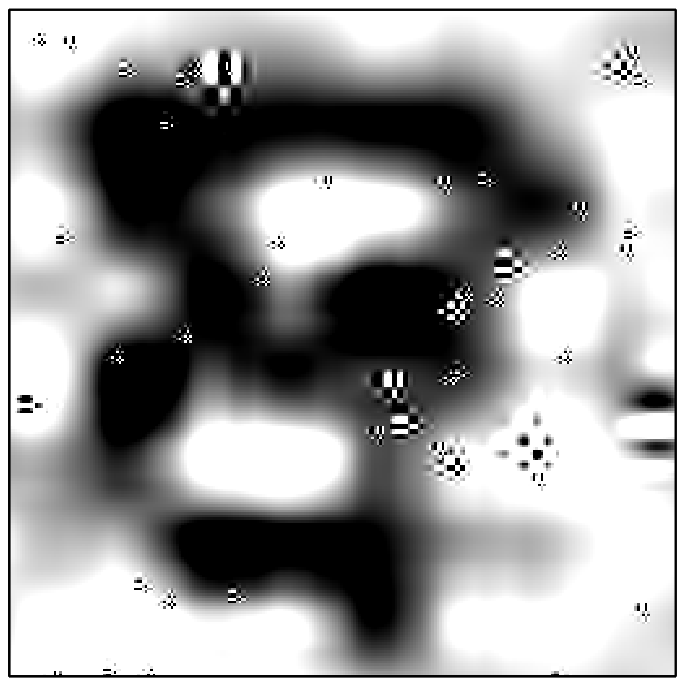} } \footnotesize\vspace{-.7cm}\centerline{$d(\widehat{\S},\S) = 162 \quad $  PSNR $= 5.9$ dB} \centerline{simulation time $0.04$s}
\\
\begin{sideways}  \hspace{1.4cm} $\mathrm{SNR} = -14$ dB \end{sideways} & 
\centerline{\includegraphics[width=2.5in]{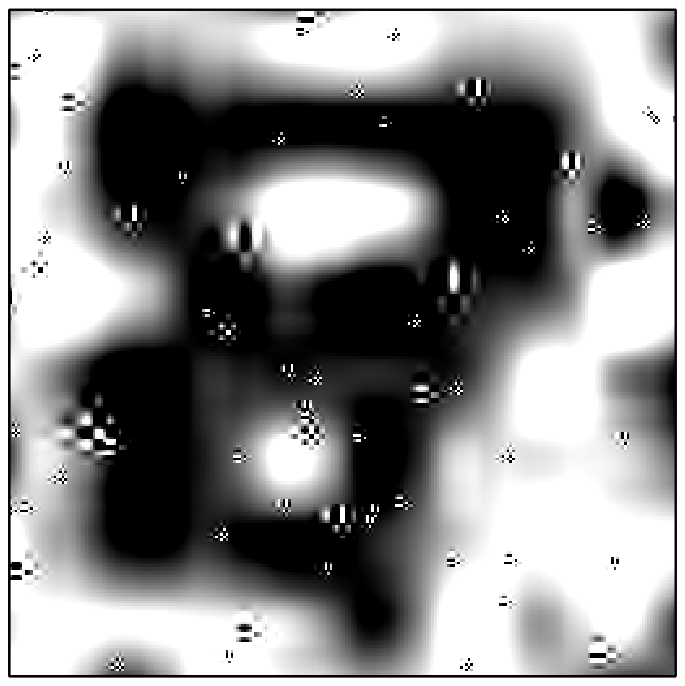} } \footnotesize\vspace{-.7cm}\centerline{$d(\widehat{\S},\S) = 232 \quad $  PSNR $= 7.29$ dB} \centerline{simulation time $0.12$s}  &   
\centerline{\includegraphics[width=2.5in]{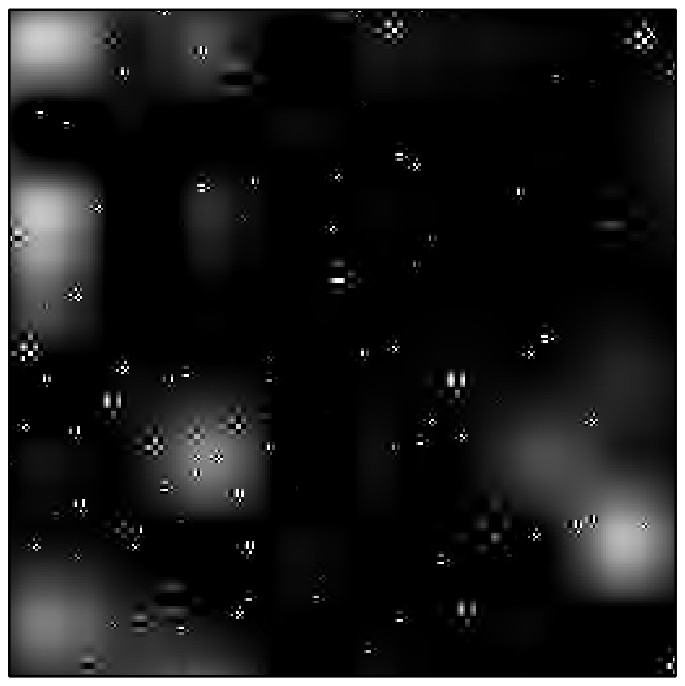} } \footnotesize\vspace{-.7cm}\centerline{$d(\widehat{\S},\S) = 380 \quad $  PSNR $= 2.67$ dB} \centerline{simulation time $94$s} &  
\centerline{\includegraphics[width=2.5in]{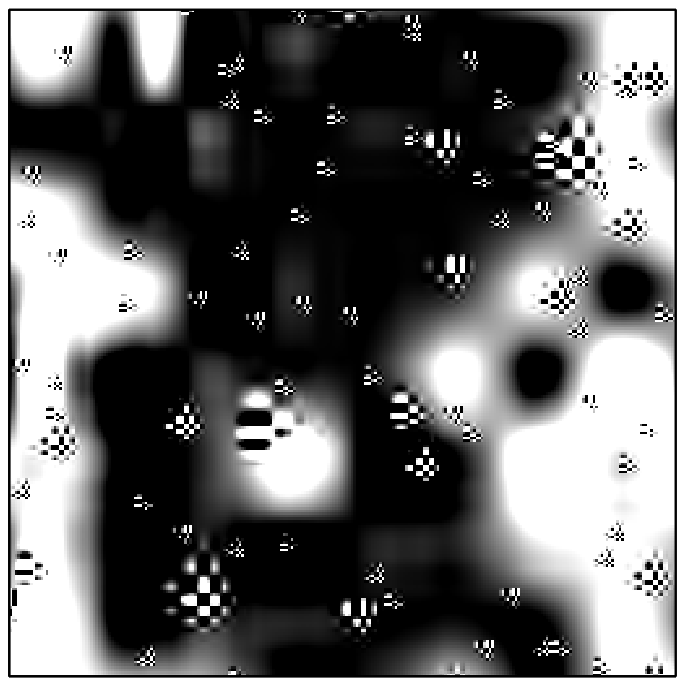} } \footnotesize\vspace{-.7cm}\centerline{$d(\widehat{\S},\S) = 310 \quad $  PSNR $= -1.94$ dB} \centerline{simulation time $0.04$s}
\\
\end{tabular}} 
\caption{
Handwritten eight \cite{FrankAsuncion:2010}, approximately sparse in `dB4' wavelet basis, $n= 256^2$.  Recovery with CASS (Alg. \ref{alg:CASS1pass}), $k =  256$, $\epsilon =1$, as input, {$m = 4096$} (column 1). Traditional compressed sensing with Gaussian ensemble and LASSO (using SpaRSA \cite{wright2009sparse}), tuned to return $k = 256$ components, {$m = 4096$} (column 2).  Direct sensing, $m=n$, using best $k=256$ term approximation (column 3).   Each method evaluated for three SNRs, $\mathrm{SNR} (\mbox{dB}) = 10 \log_{10} (x_{(k)}^2 M/n)$, where $x_{(k)}$ is the amplitude of the $k$th largest component. \label{fig:eights}}
\end{figure*}

\begin{figure*}[]
\centerline{ 
\begin{tabular}{p{.7cm}p{5.2cm}p{5.2cm}p{5.2cm}}
& \centerline{CASS (Alg. \ref{alg:CASS1pass})}  & \centerline{CS (LASSO)} & \centerline{Direct} \\ 
\begin{sideways}  \hspace{1.7cm} $\mathrm{SNR} = 16$ dB \end{sideways} & 
\centerline{\includegraphics[width=2.5in]{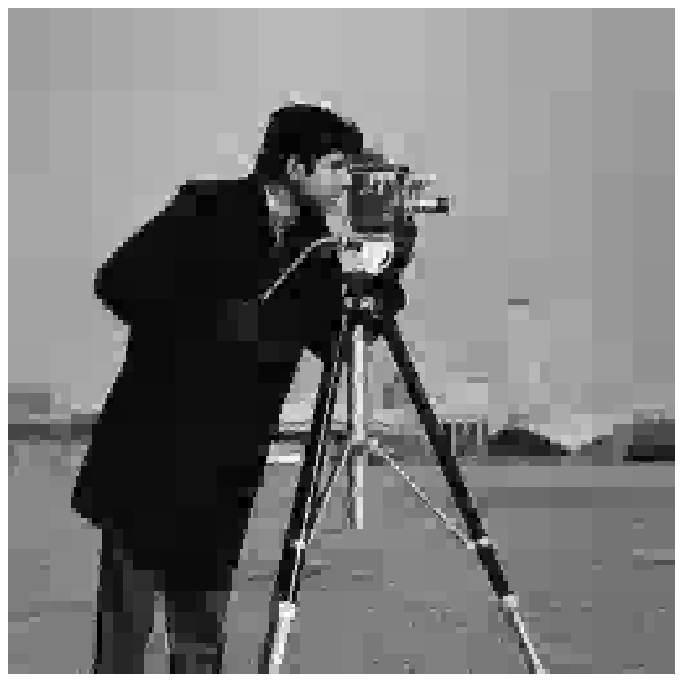} } \footnotesize\vspace{-.7cm}\centerline{$d(\widehat{\S},\S) = 736 \quad $  PSNR $= 25.5$ dB} \centerline{simulation time $0.35$s}  &   
\centerline{\includegraphics[width=2.5in]{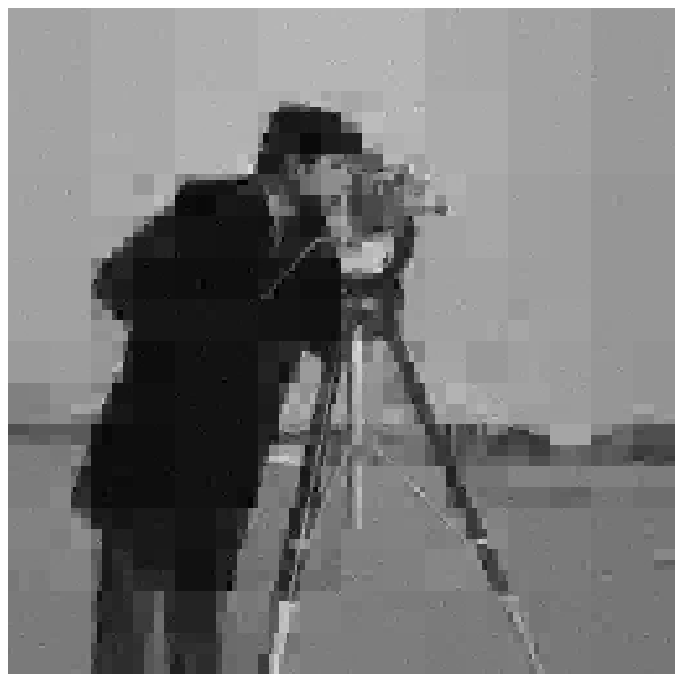} } \footnotesize\vspace{-.7cm}\centerline{$d(\widehat{\S},\S) = 1504 \quad $  PSNR $= 21.8$ dB} \centerline{simulation time $312$s} &  
\centerline{\includegraphics[width=2.5in]{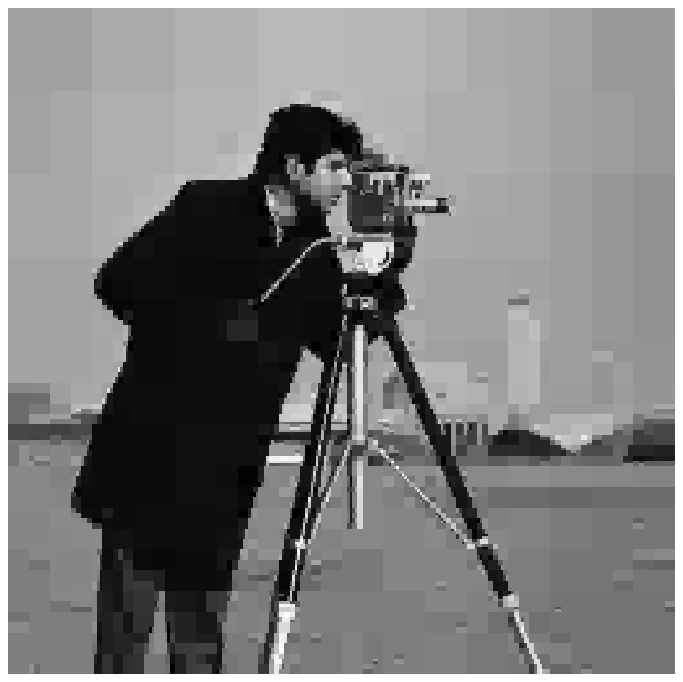} } \footnotesize\vspace{-.7cm}\centerline{$d(\widehat{\S},\S) = 316 \quad $  PSNR $= 26.2$ dB}  \centerline{simulation time $0.06$s}
\\
\begin{sideways}  \hspace{1.7cm}  $\mathrm{SNR} = 10$ dB \end{sideways}  & 
\centerline{\includegraphics[width=2.5in]{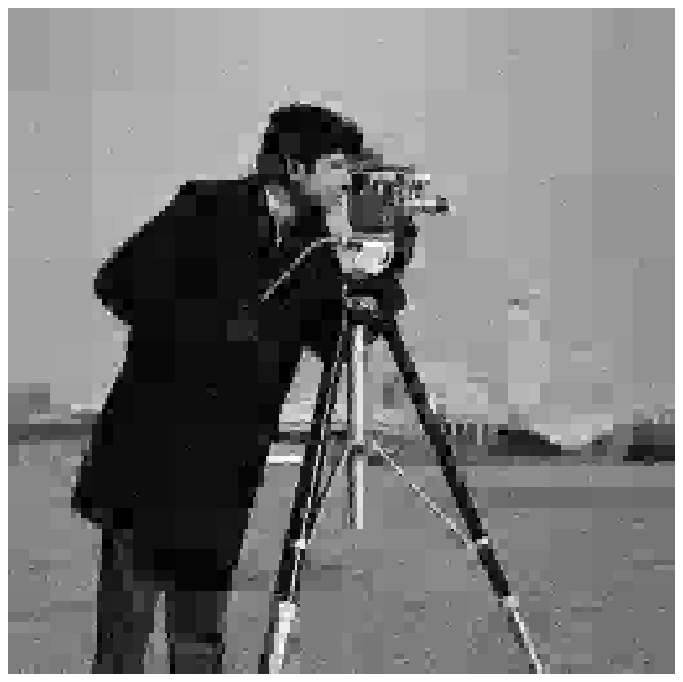} }  \footnotesize\vspace{-.7cm}\centerline{$d(\widehat{\S},\S) = 1100 \quad $  PSNR $= 24.6$ dB} \centerline{simulation time $0.34$s} &    
\centerline{\includegraphics[width=2.5in]{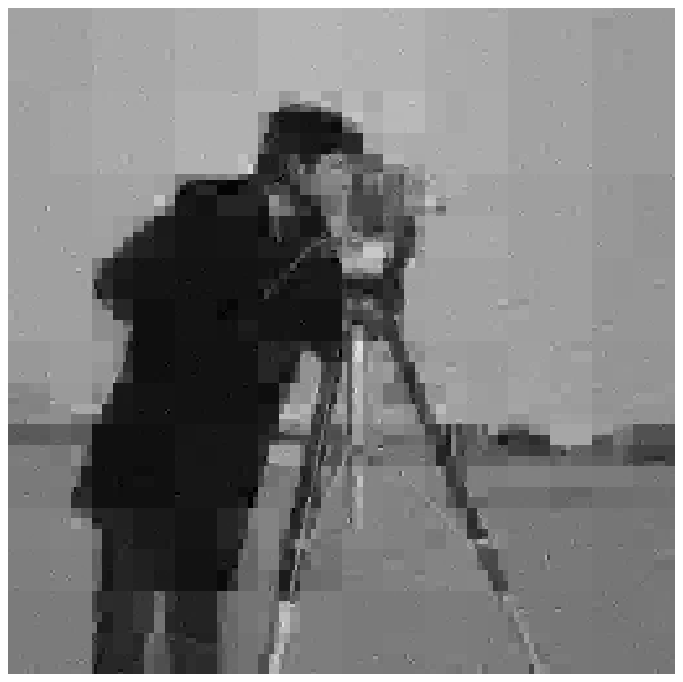} } \footnotesize\vspace{-.7cm}\centerline{$d(\widehat{\S},\S) = 1774 \quad $  PSNR $= 21.3$ dB} \centerline{simulation time $327$s} & 
\centerline{\includegraphics[width=2.5in]{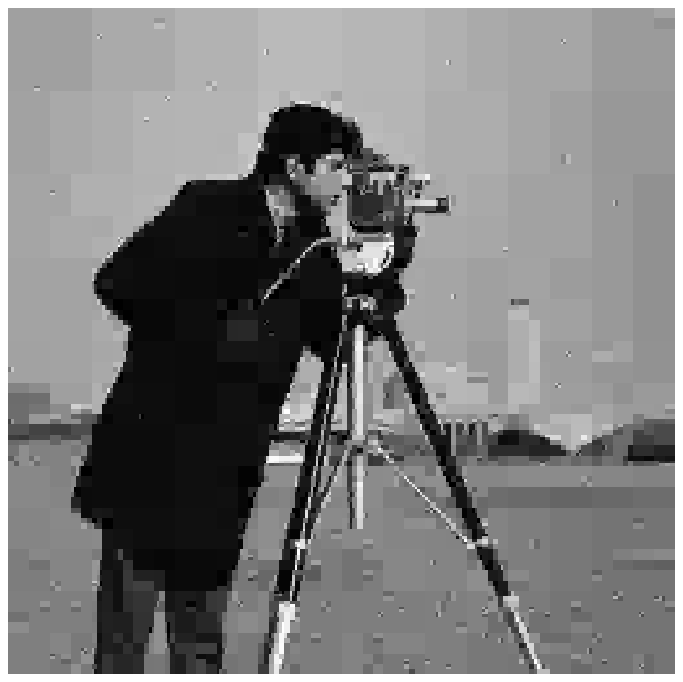} } \footnotesize\vspace{-.7cm}\centerline{$d(\widehat{\S},\S) = 714 \quad $  PSNR $=  25.3$ dB} \centerline{simulation time $0.06$s}
\\
\begin{sideways}  \hspace{1.7cm}  $\mathrm{SNR} = 4$ dB \end{sideways}  & 
\centerline{\includegraphics[width=2.5in]{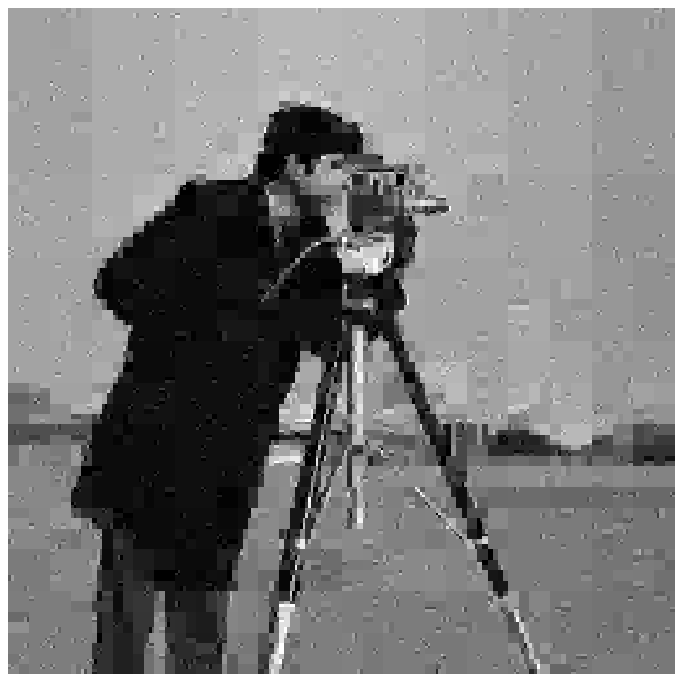}} \footnotesize\vspace{-.7cm}\centerline{$d(\widehat{\S},\S) = 1836 \quad $  PSNR $=22.7$ dB} \centerline{simulation time $0.36$s}  &     
\centerline{\includegraphics[width=2.5in]{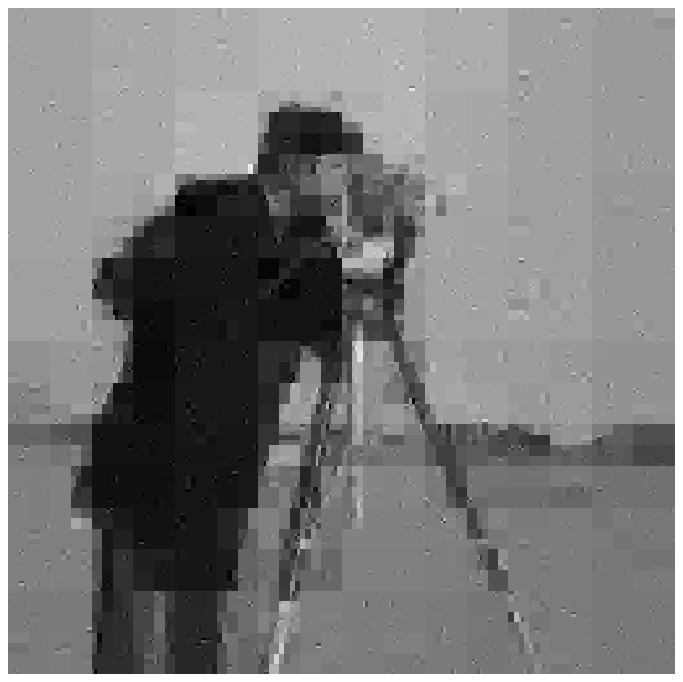} } \footnotesize\vspace{-.7cm}\centerline{$d(\widehat{\S},\S) = 2304 \quad $  PSNR $= 20.5$ dB}   \centerline{simulation time $232$s} & 
\centerline{\includegraphics[width=2.5in]{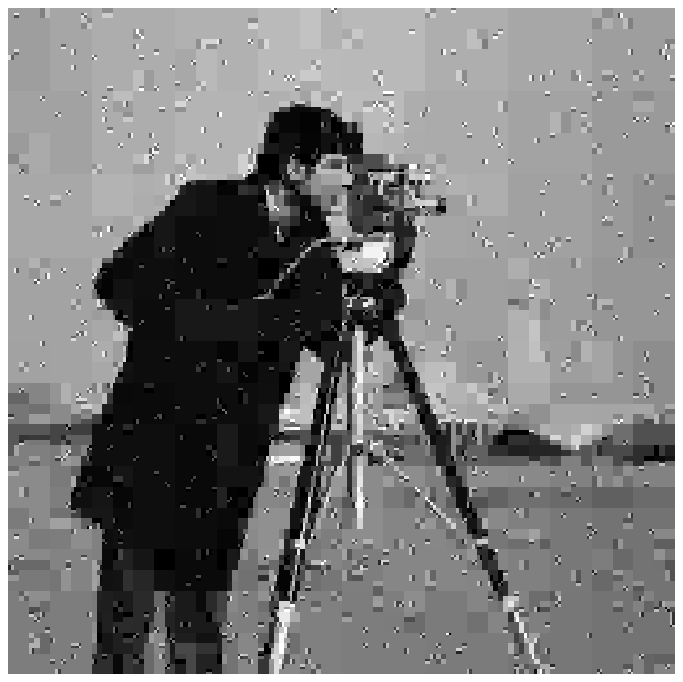} } \footnotesize\vspace{-.7cm}\centerline{$d(\widehat{\S},\S) = 1694 \quad $  PSNR $= 21.4$ dB} \centerline{simulation time $0.05$s}
\\
\end{tabular}} 
\caption{Cameraman, approximately sparse in Haar basis, $n = 256^2$.  Recovery with CASS, Alg. \ref{alg:CASS1pass},  $k = 2048$, $\epsilon = 1$, as input, $m = 20480$ (column 1).  Traditional compressed sensing with Gaussian ensemble and recovery with LASSO  (using SpaRSA \cite{wright2009sparse}), regularizer tuned to return $k = 2048$ components, $m = 20480$ (column 2).  Direct sensing, $m=n$, using best $k=2048$ term approximation (column 3).   Each method evaluated for three SNRs, $\mathrm{SNR} (\mbox{dB}) = 10 \log_{10} (x_{(k)}^2 M/n)$, where $x_{(k)}$ is the amplitude of the $k$th largest component.   \label{fig:CameraMan2}}
\end{figure*}

\vspace{.6cm}
\subsection{Discussion}

The leading constant in the upper bound -- a factor of 20 required by the CASS procedure --  is significant in practice. 
 Roughly speaking, traditional direct sensing requires $x_{\min} \geq \sqrt{ 2 \frac{n}{M} \log n}$, while CASS requires $x_{\min} \geq \sqrt{ 20 \frac{n}{M} \log k}$.  One would then conclude that anytime the signal is sufficiently sparse, specifically, if $k \leq n^{{1}/{10}}$, the CASS procedure should outperform direct sensing.  Notice that in simulation CASS outperforms direct sensing on signals where $k \geq n^{1/10}$, implying the leading constant of 20 is loose.  
\vspace{1.3cm}

\section{Conclusion}	
This work presented an adaptive compressed sensing and group testing procedure for recovery of sparse signals in additive Gaussian noise, termed CASS. The procedure was shown to be near-optimal in that it comes within a constant factor of the best possible dependence on SNR, while requiring on the order of $k \log n$ measurements.  Simulation showed that the theoretical developments hold in practice, including scenarios in which the signals is only approximately sparse. 

While many of the questions regarding the gains of adaptivity for sparse support recovery have been settled, small questions remain. {First, it would be of practical interest to further reduce the leading constant in the SNR requirement of CASS.} The question of whether or not \emph{exact} support recovery is possible in full generality (positive and negative non-zero entries, any level of sparsity), with SNR scaling as $\log k$, and the number of measurements scaling as  $k \log n$, is still outstanding.   {Other questions, such as the applicability of CASS to signals which are sparse in overcomplete dictionaries, are also interesting future research directions. }
\vspace{3.1cm}

\section*{Appendix A}
\emph{Proof of Theorem \ref{thm:kCBS}:}
The total number of measurements is given as 
\begin{eqnarray} \nonumber
m &=& \ell_0 + 2k ( s_0- 1 )  = 2k \log_2 \left( \frac{n}{k} \right)
\end{eqnarray}
since the first step requires $\ell_0$ measurements, and the remaining $s_0-1$ steps each require $2k$ measurements.  The equality follows as $\ell_0 = 4k$.
The constraint in (\ref{eqn:sensEngr}) is satisfied:
\begin{eqnarray} \nonumber
||A||_{\mathrm{fro}}^2 &=&  \sum_{\ell=1}^{\ell_0} ||\bm{a}_{1,\ell}||^2 +  \sum_{s=2}^{s_0} \sum_{\ell=1}^{2k}  ||\bm{a}_{s,\ell}||^2 \\ \nonumber
& = & \ell_0 \left(\frac{n}{\ell_0}\right) \left( \frac{M}{\gamma n} \right) + 2k \sum_{s=2}^{s_0}  \left(\frac{n }{\ell_0 2^{s-1}}\right) \left(\frac{Ms}{\gamma n} \right) \\
&=&  \frac{M}{\gamma}\left(1 +  \frac{4k}{\ell_0}\sum_{s=2}^{s_0} {s2^{-s} }\right) = M
\end{eqnarray}
where the second equality follows as each sensing vector has support $n2^{-(s-1)}/\ell_0$ and amplitude $Ms/(\gamma n)$ on step $s$.

\begin{figure*}[!t]
\normalsize
\newcounter{MYtempeqncnt}
\setcounter{MYtempeqncnt}{\value{equation}}
\setcounter{equation}{5}
\begin{eqnarray} \nonumber
\hspace{-.8cm} \P\left(\widehat{\S} \neq \S \right) 
&\leq& \sum_{j=1}^{t_1}  \P \left( |y_{1,j}| \leq \frac{x_{\min} a_1}{2}\right)  + \sum_{j=t_1+1}^{4k} \P \left( |y_{1,j}| \geq \frac{x_{\min} a_1}{2} \right)   
\\  \label{eqn:long1}
&& \qquad \qquad \qquad 
+\sum_{s=2}^{s_0}  \left(\sum_{j=1}^{t_s}  \P \left( |y_{s,j}| \leq \frac{x_{\min} a_s}{2}\right)  \right.  +\left.
\sum_{j=t_s+1}^{2k} \P \left( |y_{s,j}| \geq \frac{x_{\min} a_s}{2} \right)\right)
\\
 \label{eqn:long2}
 &\leq & 4k \exp \left( -\frac{x_{\min}^2 a_1^2}{8} \right)  + \sum_{s=2}^{s_0}  2k    \exp \left( -\frac{x_{\min}^2 a_s^2}{8} \right)  
\\
\label{eqn:long3}
&\leq&  \sum_{s=1}^{s_0} \exp\left(-\frac{x_{\min}^2  Ms }{8 \gamma n} + \log 4 k\right) 
\leq  \sum_{s=1}^{s_0} \exp\left(-\frac{x_{\min}^2  Ms }{20n} + \log 4 k\right) 
\end{eqnarray}
\setcounter{equation}{\value{MYtempeqncnt}}
\addtocounter{equation}{3}
\hrulefill
\vspace*{4pt}
\end{figure*}

A sufficient conditions for exact recovery of the support set is that the procedure never incorrectly eliminate a measurement corresponding to a non-zero entry on any of the $s_0$ steps.    Let $y_{s,1},\dots,y_{s,t_s}$ be measurements corresponding to sensing vectors with one or more non-zero entries in their support, and $y_{s,t_s+1},\dots$ be measurements corresponding to the zero entries on step $s$. Consider a series of thresholds denoted $\tau_s$, $s = 1,\dots, s_0$.  If  $|y_{s,1}|,\dots,|y_{s,t_s}|$ are all greater than $\tau_s$, and  $|y_{s,t_s+1}|,\dots$ are all less then $\tau_s$, for all $s$, the sufficient condition is implied.  Taking the complement of this event,  
and from a union bound,
\begin{eqnarray} \nonumber
\nonumber
\P(\widehat{\S} \neq \S) &\leq&   \P \left(\bigcup_{j=1}^{t_1} \left\{|y_{1,j}| \leq \tau_1 \right\} \cup \bigcup_{j = t_1+1}^{\ell_0} \left\{ | y_{1,j}| \geq \tau_1 \right\} \right) \\ \nonumber
 && \hspace{-1cm} +  \sum_{s=2}^{s_0}   \P \left(\bigcup_{j=1}^{t_s} \left\{|y_{s,j}| \leq \tau_s \right\} \cup \bigcup_{j= t_s+1}^{2k} \left\{ | y_{s,j}| \geq \tau_s \right\} \right)
\end{eqnarray}
for any series of thresholds $\tau_s$.  As all non-zero elements are positive by assumption, and the signals combine constructively, it is straightforward to see $y_{s,1},\dots,y_{s,t_s}$ are normally distributed with mean greater than $x_{\min} a_s$ and unit variance.  On the other hand, $y_{s,t_s+1},\dots$ are normally distributed with mean zero and unit variance.  Setting the thresholds at $\tau_s = x_{\min} a_s/2$, where $a_s$ is the amplitude of the non-zero entries of the sensing vector on step $s$, gives the series of equations in (\ref{eqn:long1}) - (\ref{eqn:long3}).

Here, (\ref{eqn:long1}) follows from the union bound. Equation (\ref{eqn:long2}) follows from a standard Gaussian tail bound: $1-F_{\mathcal{N}}(x) \leq\frac{1}{2} \exp (-x^2/2)$ for $x\geq 0$, where $F_{\mathcal{N}}(\cdot)$ is the Gaussian cumulative density function. Lastly, (\ref{eqn:long3}) follows as $\gamma$ is always less than $5/2$:
\begin{eqnarray} \label{eqn:gammabnd}
\gamma = 1 + 4k/\ell_0 \sum_{s=2}^{s_0} s 2^{-s} \leq \frac{5}{2}
\end{eqnarray}
as the sum is arithmetico-geometric, and $\ell = 4k$.   Setting 
\begin{eqnarray} \nonumber
x_{\min} &\geq& \sqrt{20 \frac{n}{M} \left(  \log k  + \log \left( \frac{8}{\delta}\right)  \right)} 
\end{eqnarray}
in (\ref{eqn:long3}) gives $\mathbb{P}(\widehat{\S} \neq \S) \leq  \sum_{s=1}^{s_0}  \left(2/\delta \right)^{-s}  \leq \delta$, completing the proof.
\begin{eqnarray}
\quad \nonumber
\end{eqnarray}

\section*{Appendix B}
\emph{Proof of Theorem \ref{thm:PosNeg}.}

First, we confirm that the sensing budget is satisfied.  In the same manner as before  
\begin{eqnarray} \nonumber
||A||_{\mathrm{fro}}^2 & =&  \sum_{\ell=1}^{\ell_0}  ||\bm{a}_{1,\ell}||^2  +  \sum_{s=2}^{s_0} \sum_{\ell=1}^{2k}  ||\bm{a}_{s,\ell}||^2   \\ \nonumber
&=&  \frac{M}{\gamma}\left(1 +  \frac{4k}{\ell_0}\sum_{s=2}^{s_0} {s2^{-s} }\right) = M .
\end{eqnarray}
 The total number of measurement satisfies 
\begin{eqnarray} \nonumber
m &=& \ell_0 + 2 k (s_0 - 1) \leq 8k/\epsilon + 2 k \log_2 (n/k)
\end{eqnarray}
where the inequality follows as {$\min \{n, 4 k/\epsilon\} \leq \ell_0 \leq 4 \cdot 2^{\lceil \log_2(k/\epsilon) \rceil} \leq 8 k/\epsilon$}.

We proceed by bounding the expected number of non-zero components that are isolated when the signal is partitioned, and then show these indices are returned in $\widehat{\S}$ with high probability.  We assume the support of the signal is chosen uniformly at random from all possible $k$-sparse supports sets (note this is equivalent to randomly permuting the labels of all indices).    Conceptually, the algorithm divides the signal support into $\ell_0$ disjoint support sets, where $\ell_0 = \min\{n, \mbox{next power of 2 greater than } 4k/\epsilon\}$.  
Let the set $\mathcal{I} \subset \{1,\dots,\ell_0 \}$ index the disjoint support sets that contain exactly one index corresponding to a non-zero entry:
\begin{eqnarray} \nonumber
\mathcal{I} := \left\{  \ell \in \{1,\dots,\ell_0\} : |\mathcal{J}_{\log_2{\ell_0},\ell }\cap \S|=1 \right\}.
%
\end{eqnarray}
Then
\begin{eqnarray}  \nonumber
\E[|\mathcal{I}|] &=&   \E\left[\sum_{\ell=1}^{\ell_0} \mathrm{I}_{\{\ell \in \mathcal{I} \}} \right]  =\sum_{\ell=1}^{\ell_0} \E\left[  \mathrm{I}_{\{\ell \in \mathcal{I} \}} \right] \\ \nonumber
&=& \sum_{\ell = 1}^{\ell_0} \P {\left( \ell \in \mathcal{I} \right)}  \nonumber  =\sum_{\ell=1}^{\ell_0} \; \frac{k}{\ell_0} \left(\frac{\ell_0-1}{ \ell_0} \right)^{k-1} \\ \nonumber
& =& k \left(\frac{\ell_0-1}{\ell_0} \right)^{k-1} 
\geq  k \frac{\ell_0 - k}{\ell_0}  \nonumber
\end{eqnarray}
where  $\mathrm{I}_{\{\cdot\}}$ is the indicator function. Since $\ell_0 \geq \min \{ 4k/\epsilon,n\}$ we have
\begin{eqnarray} \label{eqn:expbnd}
\E[|\mathcal{I}| ] \geq k(1-\epsilon/4).
\end{eqnarray}
Let $\mathcal{E}$ be the event that the procedure fails to find at least as many true non-zero elements as the number that were isolated on the first partitioning, specifically,  
\begin{eqnarray} \nonumber
\mathcal{E} = \{ |\widehat{\S} \cap \S | < |\mathcal{I}|  \}.
\end{eqnarray}
Additionally, notice that since the procedure returns exactly $k$ indices,
\begin{eqnarray} \label{eqn:Eimp}
\mathcal{E}^c\; \Rightarrow \; d(\widehat{\S},\S) \leq 2(k-|\mathcal{I}|).
\end{eqnarray}

We can bound the probability of the event $\mathcal{E}$ in a fashion similar to the previous proof.  Consider a threshold at step $s$ given as $\tau_s = x_{\min} a_s/2$.   Assume that $y_{s,1}, y_{s,2},\dots, y_{s,t'_s}$ are measurements corresponding to isolated components, $y_{s,t'_s+ 1},\dots, y_{s,t_s}$ correspond to measurements that contain more than one non-zero index, and  $y_{s,t_s+1},\dots$ correspond to noise only measurements.  If all noise only measurements ($y_{s,t_s+1},\dots$) are below the threshold in absolute value, and all isolated non-zero components are above the threshold ($y_{s,1}, \dots, y_{s,t'_s}$), for all $s$, this implies $\mathcal{E}$ does not occur, regardless of $y_{s,t'_s+1}, \dots, y_{s,t_s}$.  As before, after applying a union bound,
\begin{eqnarray}  \nonumber
&\P \left( \mathcal{E} \right) &\leq \sum_{j=1}^{t'_1} \P\left(|y_{1,j}| \leq \tau_1 \right) + \sum_{j=t_1+1}^{\ell_0} \P\left(|y_{1,j} | \geq \tau_1 \right)  \\ \nonumber
&  + &\sum_{s=2}^{s_0} \left(\sum_{j=1}^{t'_s} \P\left(|y_{s,j}| \leq \tau_s \right) + \sum_{j=t_s+1}^{2 k} \P\left(|y_{s,j} | \geq \tau_s \right) \right) \\ \nonumber
&\leq& \ell_0 \exp \left(- \frac{x_{\min}^2 a_1^2}{8}  \right) + \sum_{s=2}^{s_0} 2 k \exp \left(- \frac{x_{\min}^2 a_s^2}{8}  \right) \\ \nonumber
&\leq &  \sum_{s=1}^{s_0}  \exp \left(- \frac{x_{\min}^2 a_s^2}{8} + \log \left(  \frac{8k}{\epsilon}\right)  \right) 
\end{eqnarray}
where the inequality follows from a standard Gaussian tail bound, and the last line follows as $\ell_0 \leq 8k/\epsilon$.  Setting $a_s = \sqrt{Ms/\gamma n}$ gives
\begin{eqnarray} \nonumber
\P \left( \mathcal{E} \right)  &\leq& \sum_{s=1}^{s_0}  \exp \left(- \frac{x_{\min}^2 M s}{8 \gamma n} + \log \left(  \frac{8k}{\epsilon}\right)  \right)  \\ \nonumber
&\leq& \sum_{s=1}^{s_0}  \exp \left(- \frac{x_{\min}^2 M s}{20 n} + \log \left(  \frac{8k}{\epsilon}\right)  \right)
\end{eqnarray}
since, from (\ref{eqn:gammabnd}),  $\gamma \leq 5/2$.
Next, imposing the condition of the Theorem, 
\begin{eqnarray} \nonumber
x_{\min} &\geq& \sqrt{ 20 \frac{n}{M} \left( \log k   +   2\log \left(\frac{8}{\epsilon}\right)  \right)  }   \\ \nonumber
& =& \sqrt{20 \frac{n}{M} \left(  \log \left( \frac{8}{\epsilon} \right) +  \log \left(  \frac{8k}{\epsilon}\right) \right) } , 
\end{eqnarray}
 then
\begin{eqnarray} \label{eqn:PeError}
\P \left( \mathcal{E} \right) &\leq& \sum_{s=1}^{s_0} \left(\frac{8}{\epsilon} \right)^{-s} \leq \epsilon/4.
\end{eqnarray}
By combining (\ref{eqn:expbnd}),  (\ref{eqn:Eimp}), and (\ref{eqn:PeError}), since $ \E[d(\widehat{\S},\S)| \mathcal{E}] \leq 2k$, 
\begin{eqnarray} \nonumber
\E[d(\widehat{\S},\S)]  &=& \P( \mathcal{E} ) \E[d(\widehat{\S},\S)| \mathcal{E}] + \P( \mathcal{E}^c )\E[d(\widehat{\S},\S) | \mathcal{E}^c] \\
&\leq& k \epsilon, \nonumber
\end{eqnarray}
completing the proof. 

\vspace{1.8cm}

\bibliographystyle{IEEEtran}
\bibliography{kSparseCBS.bib}

\newpage

\begin{IEEEbiographynophoto}{Matthew L. Malloy} received the B.S. degree in Electrical and Computer Engineering from the University of Wisconsin in 2004, the M.S. degree from Stanford University in 2005 in electrical engineering, and the Ph.D. degree from the University of Wisconsin in December 2012 in electrical engineering.   

Dr. Malloy currently holds a postdoctoral research position at the University of Wisconsin in the Wisconsin Institutes for Discovery.  From 2005-2008, he was a radio frequency design engineering for Motorola in Chicago, IL, USA.  In 2008 he received the Wisconsin Distinguished Graduate fellowship. In 2009, 2010 and 2011 he received the Innovative Signal Analysis fellowship.  

Dr. Malloy has served as a reviewer for the IEEE Transactions on Signal Processing, the IEEE Transactions on Information Theory, IEEE Transactions on Automatic Control and the Annals of Statistics.  He was awarded the best student paper award at the Asilomar Conference on Signals and Systems (2011).   His current research interests include signal processing, estimation and detection, information theory, statistics and optimization, with applications in communications, biology, and imaging. \vspace{-.8cm}
\end{IEEEbiographynophoto}

\begin{IEEEbiographynophoto}{Robert D. Nowak} received the B.S., M.S., and Ph.D. degrees in electrical engineering from the University of Wisconsin-Madison in 1990, 1992, and 1995, respectively.  He was a Postdoctoral Fellow at Rice University in 1995-1996, an Assistant Professor at Michigan State University from 1996-1999,  held Assistant and Associate Professor positions at Rice University from 1999-2003, and is now the McFarland-Bascom Professor of Engineering at the University of Wisconsin-Madison. 
 \newpage
 
Professor Nowak has held visiting positions at INRIA, Sophia-Antipolis (2001), and Trinity College, Cambridge (2010). He has served as an Associate Editor for the IEEE Transactions on Image Processing and the ACM Transactions on Sensor Networks, and as the Secretary of the SIAM Activity Group on Imaging Science. He was General Chair for the 2007 IEEE Statistical Signal Processing workshop and Technical Program Chair for the 2003 IEEE Statistical Signal Processing Workshop and the 2004 IEEE/ACM International Symposium on Information Processing in Sensor Networks. 

Professor Nowak received the General Electric Genius of Invention Award (1993), the National Science Foundation CAREER Award (1997), the Army Research Office Young Investigator Program  Award (1999), the Office of Naval Research Young Investigator Program Award (2000), the IEEE Signal Processing Society Young Author Best Paper Award (2000), the IEEE Signal Processing Society Best Paper Award (2011), ASPRS Talbert Abrams Paper Award (2012), and the IEEE W.R.G. Baker Award (2014). He is a Fellow of the Institute of Electrical and Electronics Engineers (IEEE). His research interests include signal processing, machine learning, imaging and network science, and applications in communications, bioimaging, and systems biology.
\end{IEEEbiographynophoto}

\end{document}